\RequirePackage{fix-cm}
\documentclass[smallextended]{svjour3}       % onecolumn (second format)
\smartqed  % flush right qed marks, e.g. at end of proof
\usepackage{amssymb}
\usepackage{amsmath}
\usepackage{graphicx}
\usepackage{color}
\usepackage[dvipsnames]{xcolor}
\usepackage{url}
\usepackage{natbib}
\usepackage{ulem}

\graphicspath{./figures}

\newcommand{\vect}[1]{\boldsymbol{{#1}}}
\newcommand{\E}{\mathbb{E}}
\newcommand{\Var}{\mathrm{Var}}

\renewcommand{\d}{\mathrm{d}}
\newcommand{\nbb}{\mathbb{N}}
\usepackage{amsopn}

\DeclareMathOperator{\Ran}{Ran}
\newcommand{\given}{\, | \,}

\newcommand{\defdista}{\stackrel{\mathcal{D}}{=}}

\renewcommand{\emph}[1]{\textit{#1}}

\begin{document}

\title{Renewal reward perspective on linear switching diffusion systems\thanks{MVC is supported by The Ohio State University President's Postdoctoral Scholars Program and by the Mathematical Biosciences Institute at The Ohio State University through NSF DMS-1440386. JF, PRK, and SAM are supported by NIH R01GM122082-01.}
}


\author{Maria-Veronica Ciocanel  \and
        John Fricks \and
        Peter R.~Kramer \and
        Scott A.~McKinley %etc.
}


\institute{Maria-Veronica Ciocanel \at
              Mathematical Biosciences Institute, The Ohio State University\\
              Tel.: +614-688-3334 \\
              \email{ciocanel.1@mbi.osu.edu}      
           \and
           John Fricks \at
              School of Mathematical and Statistical Sciences, Arizona State University 
           \and
           Peter R.~Kramer \at
              Department of Mathematical Sciences, Rensselaer Polytechnic Institute  
           \and
           Scott A.~McKinley \at
              Department of Mathematics, Tulane University
              }

\date{Received: date / Accepted: date}

\maketitle

\begin{abstract}
In many biological systems, the movement of individual agents of interest is commonly characterized as having multiple qualitatively distinct behaviors that arise from various biophysical states. This is true for vesicles in intracellular transport, micro-organisms like bacteria, or animals moving within and responding to their environment. For example, in cells the movement of vesicles, organelles and other intracellular cargo are affected by their binding to and unbinding from cytoskeletal filaments such as microtubules through molecular motor proteins. A typical goal of theoretical or numerical analysis of models of such systems is to investigate the effective transport properties and their dependence on model parameters. While the effective velocity of particles undergoing switching diffusion dynamics is often easily characterized in terms of the long-time fraction of time that particles spend in each state, the calculation of the effective diffusivity is more complicated because it cannot be expressed simply in terms of a statistical average of the particle transport state at one moment of time. However, it is common that these systems are regenerative, in the sense that they can be decomposed into independent cycles marked by returns to a base state. Using decompositions of this kind, we calculate effective transport properties by computing the moments of the dynamics within each cycle and then applying renewal-reward theory. This method provides a useful alternative large-time analysis to direct homogenization for linear advection-reaction-diffusion partial differential equation models. Moreover, it applies to a general class of semi-Markov processes and certain stochastic differential equations that arise in models of intracellular transport. Applications of the proposed renewal reward framework are illustrated for several case studies such as mRNA transport in developing oocytes and processive cargo movement by teams of molecular motor proteins.

\keywords{renewal rewards \and intracellular transport \and processive motor transport}
 \subclass{60J20 \and  60J27 \and 92B05 \and 92C37}
\end{abstract}

\section{Introduction}
Microscale biological agents frequently change biophysical state, which results in significant changes in their movement behavior. Intracellular cargo, for example, switches among active transport, diffusive transport, and paused states, each resulting from different mechanochemical configurations of the cargo, cytoskeletal filaments, and the molecular motors that bind them \citep{hancock2014bidirectional,bressloff2013stochastic}. Models for this behavior can be either deterministic (typically partial differential equations, PDEs) or stochastic (often continuous-time Markov chains, CTMCs, or stochastic differential equations, SDEs) depending on whether the investigation focuses on population properties (deterministic methods) or individual paths (stochastic methods). Each state is commonly characterized in terms of a mean velocity, fluctuations about the mean velocity, and a distribution of time spent in the state, sometimes but not always determined by classical reaction rate theory. Explicit solutions for these models are rarely available, so asymptotic or numerical methods are often deployed to investigate and characterize the model's predictions. The study of deterministic models often relies on numerical simulation using PDE integration methods \citep{wang2003robust,cox2002exponential,trong2015cortical}, while stochastic models are simulated with Monte Carlo/Gillespie algorithms \citep{muller2008tug,kunwar2010robust,muller2010bidirectional,allard2019sliding} to generate individual trajectories that are then analyzed statistically. However, these computations can be quite costly, especially when one wants to understand how bulk transport properties (like effective velocity or diffusivity) depend on individual model parameters. When possible, asymptotic analysis allows for explicit approximation of transport properties, which can validate, complement, or even replace numerical simulations \citep{reed1990approximate,brooks1999probabilistic,pavliotis2005multiscale,pavliotis2008multiscale,popovic2011stochastic,mckinley2012asymptotic,bressloff2015stochastic,ciocanel2017analysis}.

The long-term effective velocity of state-switching particles is often straightforward to compute, usually obtained by calculating the fraction of time spent in each state and correspondingly averaging the associated state velocities. On the other hand, this weighted average technique is not valid when calculating a particle's effective diffusivity, since this quantity cannot be simply expressed in terms of the statistics of the particle transport at a single moment of time. Rather the effective diffusivity depends on temporal correlations that exist in displacements, including across changes in biophysical state. Generalizing some previous work \citep{brooks1999probabilistic,hughes2011matrix,krishnan2011renewal,hughes2012kinesins,ciocanel2017analysis}, we consider this problem of computing effective diffusivity for a class of state-switching particle models that can be expressed in a framework where the sequence of states are given by a Markov chain, but the times spent in these states are not necessarily exponentially distributed as in a continuous-time Markov chain. Since we assume that the state process Markov chain is positive recurrent, the particle position process can be described as a regenerative increment process in a sense defined by \citet{serfozo2009basics}, for example. That is to say, we consider processes that almost surely return to some base state at a sequence of (random) regeneration times such that the dynamics after a regeneration time are independent from those that occur before.  As a result, we can decompose the process into what we refer to as \textit{cycles}, in which the particle starts in a base state, undergoes one or more state changes, and then revisits the base state again. The dynamics within each cycle are independent of other cycles and we can use the renewal-reward theorem to perform asymptotic calculations by viewing the total displacement within each cycle as its reward and viewing the cycle durations as times between regenerations. An early application of the idea of computing effective particle velocity and diffusivity by decomposition and analysis of the dynamics in terms of independent cycles was to large enhancement of (non-switching) particle diffusion in a tilted periodic potential~\citep{reimann2002diffusion,reimann2001giant}.

Our primary motivating examples are related to intracellular transport. Some prominent recent investigations include the study of mRNA localization in oocyte development \citep{zimyanin2008vivo,trong2015cortical,ciocanel2017analysis}, cell polarization in the budding yeast \citep{bressloff2015stochastic}, neurofilament transport along axons \citep{jung2009modeling,li2014deciphering}, interactions of teams of molecular motor proteins \citep{klumpp2005cooperative,muller2008tug,kunwar2010robust,muller2010bidirectional}, and sliding of parallel microtubules by teams of cooperative identical motors \citep{allard2019sliding}. Microtubule-based transport of cargo is typically mediated by kinesin motors moving material to the cell periphery and by dynein motors carrying it to the nucleus. Understanding population-scale behaviors, such as protein localization, that arise from local motor interactions remains an open question. While multiple motor interactions are usually thought to be resolved through a tug-of-war framework \citep{muller2008tug}, it has been observed that important predictions made by the tug-of-war framework are not consistent with \textit{in vivo} experimental observations \citep{kunwar2011mechanical,hancock2014bidirectional}. The work presented in this paper can aid theoretical efforts to relate local motor-cargo dynamics to predictions for large scale transport.

\subsection{PDE methods for Markovian switching}

For hybrid switching diffusion processes \citep{book-yinzhu}, in which particles independently switch with continuous-time Markovian dynamics between states that have different velocities and/or diffusivities, the law of a particle can be expressed in terms of its associated forward Kolmogorov equations with an advection-reaction-diffusion structure: 
\begin{equation}\label{eq:n-state-mod}
\frac{\partial \vect{u}(y,t)}{\partial t} = A^T\vect{u} - V \partial_y \vect{u} + D \Delta \vect{u}\,.
\end{equation}
We will actually think of $\vect{u}$ as an $(N+1)$-dimensional column vector (indexed from $0$ to $N$)  of the concentrations of particle populations in different dynamical states, which also obey the forward Kolmogorov equations with a different normalization. The dynamics are governed by matrices $A, V, D \in \mathbb{R}^{(N+1) \times (N+1)}$, where $ V $ and $ D $ are diagonal matrices, with real constant diagonal entries $v_0, v_1,\ldots,v_N$  for $ V $ corresponding to the particle velocities in each state, and positive real constant diagonal entries $d_0, d_1,\ldots,d_N$ for $D $ corresponding to the diffusion coefficients  in each state. The matrix $A$ is the transition rate matrix of the associated finite state continuous-time recurrent Markov chain (CTMC), $J(t)$, which tracks the state of the particle at a given time. That is to say, each off-diagonal entry $a_{ij}$ can be interpreted as the rate at which a particle in state $i$ switches to state $j$. The diagonal entries of $A$ are non-positive and correspond to the total rate out of a given state. The rows of $A$ sum to zero. Assuming that the CTMC is irreducible, it follows that $A$ admits a zero eigenvalue with algebraic and geometric multiplicity one, and the corresponding normalized zero-eigenvector $\vect{\pi}$ is the stationary distribution of $J(t)$.  

Either quasi-steady-state reduction~\citep{bressloff2013stochastic} or homogenization theory~\citep{pavliotis2008multiscale} can be used to reduce the complexity of the advection-diffusion-reaction system~\ref{eq:n-state-mod} to a scalar advection-diffusion equation of the form:
\begin{equation}\label{eq:eff-mod}
\frac{\partial c(y,t)}{\partial t} = v_{\text{eff}} \partial_y  c(y,t)+ D_{\text{eff}} \Delta c(y,t)\,.
\end{equation}
with constant effective velocity $ v_{\text{eff}} $ and constant effective diffusivity $ D_{\text{eff}} $ for the particle concentration without regard to state $ c(y,t) = \sum_{i=0}^N u_i (y,t) $.  Quasi-steady-state reduction assumes the stochastic switching dynamics occurs on a fast scale relative to the advection-diffusion dynamics, while homogenization theory applies at sufficiently large space and time scales relative to those characterizing the dynamical scales.  These different asymptotic conditions give in general distinct results when the transport and switching rates have explicit dependence on space, but when, as in the present case, they are spatially independent, the formulas for the effective transport coefficients coincide.  (This is because the time scale of advection/diffusion is linked purely to the spatial scale, so the large spatial scale assumption of homogenization will perforce induce a time scale separation between the switching and transport dynamics, as assumed in quasi-steady-state reduction.)
The effective velocity is computed by averaging the velocity in each state, weighted by the stationary distribution of the particle states:
\begin{equation}\label{eq:ds_vel}
v_{\text{eff}} = \vect{v} \cdot\vect{\pi}\,,
\end{equation}
where $\vect{v} = (v_0,v_1,\ldots,v_N)^T$. The effective diffusivity is given, from  an equivalent long-time effective dynamical description for intracellular transport derived by Lyapunov-Schmidt reduction on the low wavenumber asymptotics of the Fourier transform of Eq.~\ref{eq:n-state-mod} in~\citet{ciocanel2017analysis,ciocanel2018modeling}, by
\begin{align}\label{eq:ds_diff}
D_{\text{eff}} &= \vect{d} \cdot\vect{\pi} - \vect{v} \cdot 
     (\overline{A^T})^{-1} 
     (\vect{v}\circ \vect{\pi} -v_{\mathrm{eff}}\vect{\pi}) \,,
\end{align}
with 
\begin{equation}
\vect{v} \circ \vect{\pi} = (v(0)\pi(0),v(1)\pi(1),\ldots,v(N)\pi(N))^T \label{eq:hadprod}
\end{equation}denoting the Hadamard product (componentwise multiplication of vectors). Here $\vect{d} = (d_0,d_1,\ldots,d_N)^T$, and $ \overline{A^T} $ is the restriction of $ A^T $ to its range $ \Ran (A^T) $ (vectors orthogonal to $ (1,1,\ldots,1)^T $).  Note that the operation involving the inverse of $ (\overline{A^T})^{-1}$ is well-defined since $ \overline{A^T}$ is a full-rank matrix mapping $ \Ran(A^T) $ to $ \Ran (A^T) $, and its inversion in Eq.~\ref{eq:ds_diff} is applied to a vector in $ \Ran (A^T) $.  We remark that the homogenization formula is often written~\citep{dc:aih,pavliotis2008multiscale} in an equivalent adjoint form to Eq.~\ref{eq:ds_diff}, with a centering of the leading vector $ \vect{v} \rightarrow \vect{v} - v_{\mathrm{eff}} (1,1,\ldots,1)^T$ that renders the formula indifferent to the choice of how to invert $ A^T $. The term $\vect{d} \cdot\vect{\pi}$ above reflects the contributions to the asymptotic diffusivity from pure diffusion, while the second term captures the interactions between the advection and reaction terms. 

Applications of quasi-steady-state reduction to biophysical systems with state switching and diffusion can be found in~\citet{newby2010quasi,newby2010random,bressloff2011quasi,bressloff2013stochastic,bressloff2015stochastic}. Homogenization of Brownian motor models was conducted in~\citet{pavliotis2005multiscale,prk:tcgss}.

\subsection{Summary of method based on regeneration cycles}
These foregoing methods \citep{pavliotis2008multiscale,ciocanel2017analysis,bressloff2013stochastic} rely on the fully Markovian structure of the dynamics, with the state-switching process in particular taking the form of a continuous-time Markov chain with exponentially distributed state durations.  In this work, we consider a generalized framework in which we require only that the sequence of states visited form a discrete-time recurrent Markov chain, but do not require exponentially distributed state durations, so the state-switching process $ J(t) $ need not be a continuous-time Markov chain. 
Moreover, we allow for more general random spatial dynamics within a state that also need not be fully Markovian. Our framework and method of analysis rather requires only a regenerative structure of the dynamics, with repeated returns to a base state, at which moments the future randomness is decoupled from the previous history.

We use renewal-reward theory and a functional central limit theorem to derive effective drift and diffusion for these more general switching diffusion systems in terms of the analysis of a single regeneration cycle. The calculation framework also results in an expression for the expected run length of cargo undergoing switching diffusion. Our approach builds on previous applications of renewal-reward processes modeling motor-stepping and chemical cycles of bead-motor assays \citep{krishnan2011renewal,hughes2011matrix,hughes2012kinesins,miles2018analysis,shtylla2015mathematical} and extends the technique to accommodate more complex models with dynamics depending on the amount of time spent in the current state, as described in \S\ref{sec:exs}. Given the renewal-reward framework, the analysis of the model reduces to computing the correlated spatial displacement and time duration of each cycle, which we study in \S\ref{sec:main}. 

We illustrate the usefulness of the probabilistic renewal-reward techniques with several case studies. In \S\ref{sec:application_intra}, we show that our method of deriving effective velocity and diffusivity agrees with predictions in \citet{ciocanel2017analysis} arising from a Liapunov-Schmidt reduction approach equivalent to homogenization for partial differential equations  describing mRNA concentrations as in \eqref{eq:n-state-mod}. In \S\ref{sec:application_tugwar}, we show that our method also agrees with previous theoretical and numerical analyses of transport properties for cargo pulled by teams of molecular motors. In the case of tug-of-war dynamics, with cargo transported by teams of opposite-directed motors, our framework provides predictions on the dependence of effective diffusivity on the ratio of stall to detachment force of the pulling motors. We also apply this method to a model accounting for increased reattachment kinetics when motors are already attached to the cargo and show that teams of opposite-directed motors have lower effective velocities but larger run lengths than teams consisting of the dominant motor only. Finally, we show that our effective diffusivity calculation agrees with stochastic simulations of sliding microtubule behavior driven by teams of bidirectional motors for a large range of load sensitivity. As the experimental data on motor interactions develops rapidly, the framework proposed may prove useful in analyzing novel models and in understanding the dependence of effective transport properties on model parameters.

\section{Mathematical Framework and Examples}
\label{sec:exs}

The type of path we have in mind in this work is displayed in Figure \ref{fig:trajectory}, a continuous, stochastic process that switches between several stereotypical behaviors. Let the real-valued process $\{X(t) : t \geq 0\}$ be the time-dependent position of a particle and let $\{J(t) : t \geq 0\}$ denote the time-dependent underlying (e.g., biophysical) state, taking values from the finite state space $S=\{0, 1, 2, \ldots, N\}$. Switches between the states take place at the random times $\{t_k : k \in \mathbb{N}\}$ and we use $\{J_k : k \in \mathbb{N}\}$ to denote the state during the $k$th time interval $[t_{k-1},t_k)$. We set $t_0=0$ and $J_1=J(0)$. We assume that the sequence of states $\{J_k : k \in \mathbb{N}\}$ forms a time-homogeneous recurrent Markov chain with zero probability to remain in the same state on successive epochs. Given the state $J_k$, the associated state duration $ t_{k}-t_{k-1}$ and spatial displacement $X(t_{k}) - X(t_{k-1})$ are conditionally independent of all other random variables in the model (but not necessarily of each other). Moreover, the conditional joint distribution of $ t_{k}-t_{k-1} $ and $ X(t_{k})-X(t_{k-1}) $ given $ J_k $ depends only on the value of $J_k$ and not separately on the index $ k $.  In other words, the dynamics of $ (J(t),X(t))$ have a statistically time-homogeneous character.  

\begin{figure}[ht!]
\centering
\includegraphics[width=0.8\textwidth]{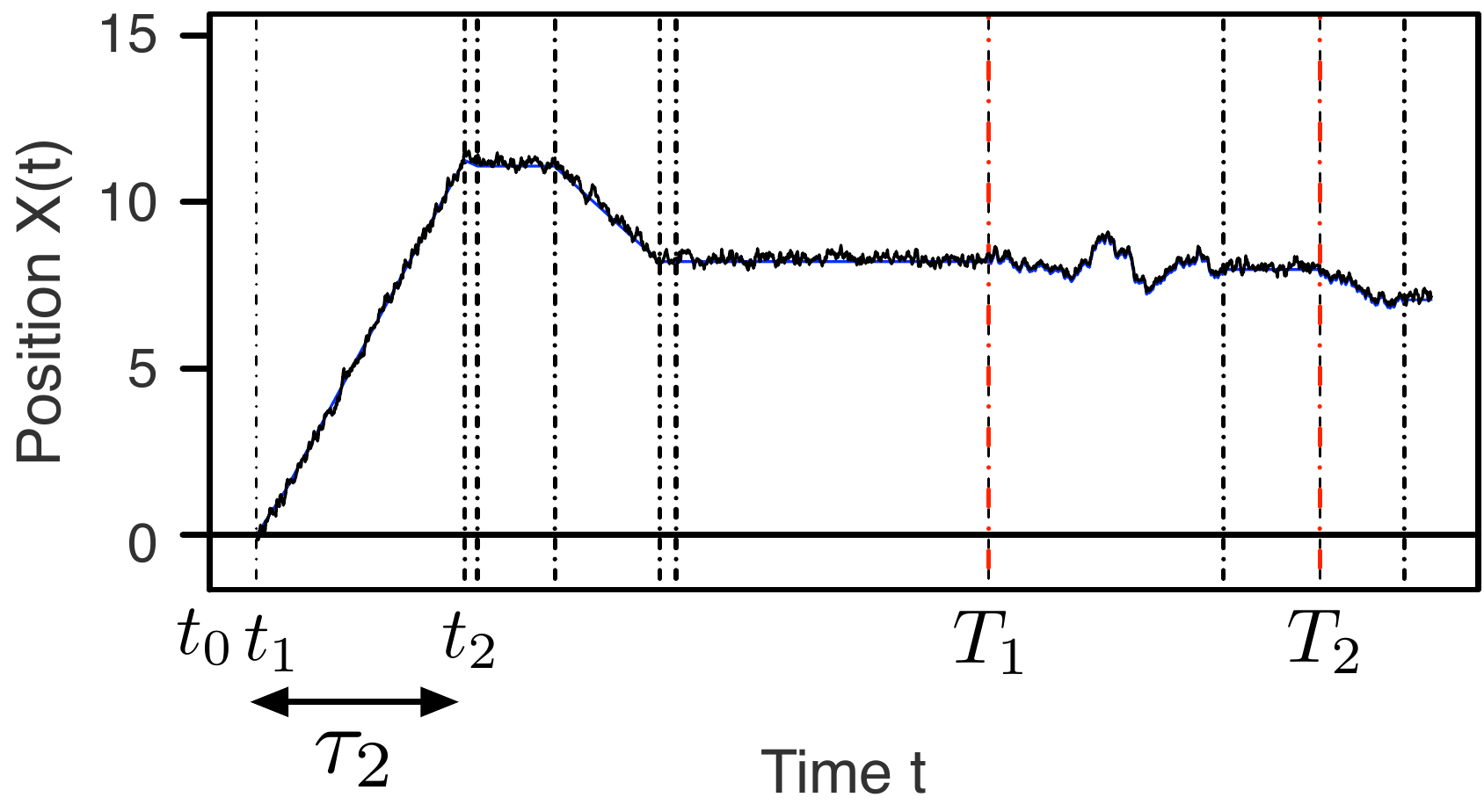} 
\caption{An example of the type of particle trajectory considered in this work. Intracellular cargo in the form of a vesicle experiences periods of active and diffusive transport. The dashed vertical lines indicate random times $\{t_k : k \in \mathbb{N}\}$ when there are switches in the biophysical state. The base ``renewal'' state is free diffusion and the red dashed vertical lines mark moments $\{T_k : k \in \mathbb{N}\}$ when the system enters the base state. We denote the times spent at each time step by $\tau_k$, as detailed in \S~\ref{sec:procedure} and \ref{sec:notation}. In the language of the paper, the red lines correspond to the regeneration times and the total spatial displacements and times between these regeneration times are the ``rewards'' and the cycle durations, respectively. \label{fig:trajectory}}
\end{figure}

One general subclass of the processes considered can be expressed as follows:  The random times $ \{t_k\}_{k=0}^{\infty} $ are generated by sampling $ t_{k}-t_{k-1} $ independently from their conditional marginal distributions given the Markov chain states $J_k$, and then conditioned upon these random variables, the spatial process $ X(t)$ is governed by a stochastic differential equation with coefficients depending on the current state, the value of $ X $ upon entry into the current state, and the time since entry into the current state.  That is, we express the conditional dynamics of $ X (t) $ as:   
\begin{equation} \label{eq:sde-example}
	\begin{aligned}
		\d X(t) = \sum_{k = 1}^\infty 1_{[t_{k-1},t_{k})} (t) \Big(\alpha_{J_k}\Big(X(t),X(t_{k-1}),t-t_{k-1}\Big) \d t + \sqrt{2 d_{J_k}} \d W(t)\Big),
	\end{aligned}
\end{equation}
where $\alpha_j : \mathbb{R}^2\times \mathbb{R}_+ \to \mathbb{R}$ is a function that describes
the drift (deterministic component of the dynamics) of a particle while in the state $j$, and $d_j$  is the diffusivity in that state. (In general, the diffusive coefficients might also depend on the position of the particle and the recent history of the process, but we restrict ourselves to memoryless, additive noise for this discussion.) In the following paragraphs, we describe a few examples.

First, we can consider the stochastic process associated with the PDE \eqref{eq:n-state-mod}, for which we set the drift terms in \eqref{eq:sde-example} to be $\alpha_j = v_j$, where $v_j$ is the constant $j$th diagonal entry of the velocity matrix $V$ in \eqref{eq:n-state-mod}. The diffusion coefficients would correspond to the entries of the diagonal matrix $D$ in \eqref{eq:n-state-mod}. The process would switch between states as a CTMC with rate matrix $A$, which, through standard theory~\citep[Sec. 3.2]{gfl:isp}, induces a transition probability matrix for the  sequence of states $ \{ J_k : k \in \mathbb{N}\}$:
\begin{equation*}
    P_{ij}= \begin{cases} \frac{A_{ij}}{\bar{\lambda}_i} & \text{ for } i \neq j, \\
    0 & \text{ for } i=j.\end{cases}
\end{equation*}
where $ \bar{\lambda}_i \equiv \sum_{j \in S \setminus \{i\}} $ is the total transition rate out of state $ i$.  The state duration in state $ i $ would be exponentially distributed with mean $ \bar{\lambda}_i^{-1} $.

For a second example, we consider the process depicted in Figure \ref{fig:trajectory}. This process could describe a model for intracellular transport, with transport states driven by processive molecular motors. There are four states: a forward processing state, a backward processing state, and a stalled state -- each of which are characterized by having a drift with speed $v_j$ plus Ornstein-Uhlenbeck type fluctuations (as described for example in \citet{smith2018assessing}) -- and a freely diffusing state where the drift term equals zero. That is, $\alpha_0 = 0$ for the freely diffusing state and for $j > 0$,
\begin{equation} \label{eq:anchored-drift}
	\alpha_j(y,y_0,t) = -\frac{\kappa}{\gamma} \big(y - (v_j t + y_0)\big)\,,
\end{equation}
where $\kappa$ is a spring constant, $\gamma$ is the viscous drag and $v_j$ is the velocity associated with the $j$th state. The term $(v_j t + y_0)$ indicates the theoretical position of a processive molecular motor that is simultaneously bound to the particle and to a microtubule. 

\begin{remark} We note that there are at least two ways that the process $X(t)$ can be considered to be non-Markovian and still fall within the set of models to which our results apply. The first, which is captured by the drift term \eqref{eq:anchored-drift}, is that the process $X(t)$ has memory in the sense that resolving $X(t)$ on any interval in $(t_k, t_{k+1})$ depends on the value $X(t_k)$. A second allowable non-Markovian dynamic can be obtained by choosing the state duration times $t_k - t_{k-1}$ given state $J_k $ to have a non-exponential distribution. As long as the stochastic process of states $\{J_k\}$ is a time-homogeneous, positive recurrent Markov chain, the technique we present will apply. 
\end{remark}

In \S\ref{sec:examples} we share a few examples from the molecular motors literature that include detailed assumptions about the set of achievable states and transitions among them. We note that these examples vary in their assumptions about fluctuations about mean behavior. In some cases, the dynamics are assumed to be ``piecewise deterministic'', similar to the class of models studied by \citet{brooks1999probabilistic} in which each state is characterized by a fixed velocity parameter $\alpha_j = v_j$ with the state diffusivity $d_j$ set to zero. 
In some of the other examples, fluctuations about the mean are included and would contribute to the long-term diffusivity as a result. Of course, fluctuations are always present in these dynamics (sometimes due to variability in the motor stepping, sometimes due to fluctuations in cargo position). There are natural ways to add these considerations to the models in \S\ref{sec:examples} and express the dynamics within the framework of equation \ref{eq:sde-example}.

\subsection{Decomposition into regenerative cycles and renewal-reward structure}\label{sec:procedure}
Here we outline our procedure for calculating the effective velocity and diffusivity of particles undergoing switching dynamics. The strategy is to break the process into independent ``cycles'' that are marked by returns to a chosen base state.  An elementary exposition of this ``dissection principle'' concept can be found in \citet[Sec. 2.5]{sr:asp}.
We define these times of re-entry into the base state as \emph{regeneration times} $\{T_n\}$. In what follows, we will view the consecutive spatial displacements and time durations of the regenerative cycles to be the rewards and cycle durations of a classical renewal-reward process \citep{cox1962renewal}. Because the cycle statistics are independent and identically distributed after the first regeneration time $ T_1$, we define (in the sense of distribution) random variables for a generic cycle $n\geq 2$:
\begin{equation} \label{eq:notation-cycle}
    \begin{aligned}
    \Delta X &\defdista X(T_{n}) - X(T_{n-1}); \qquad \Delta T \defdista T_{n}-T_{n-1}; \quad \text{ and} \\
    M &\defdista \sup_{t\in [T_{n-1},T_{n}]} |X(t) - X(T_{n-1})|.
    \end{aligned}
\end{equation}

We rely on the functional central limit theorem (FCLT) presented in \citet{serfozo2009basics} for our asymptotic results. To this end, we define the quantities
\begin{equation}
    \begin{aligned}
    \mu &:= \E(\Delta T); \quad a := \frac{\E(\Delta X)}{\E(\Delta T)}; \text{ and} \\
    \sigma^2 &:= \text{Var}(\Delta X - a \Delta T).
    \end{aligned}
\end{equation}
As in previous work on molecular motor systems \citep{hughes2011matrix,hughes2012kinesins}, the FCLT justifies defining the effective velocity and effective diffusivity of the process $X(t)$ in terms of properties of the regenerative increments as follows:
\begin{align}
    v_{\text{eff}} &:= \lim_{t \to \infty} \frac{1}{t} X(t) = a = \frac{\E(\Delta X)}{\E(\Delta T)}; \label{eq:eff_speed_serfozo} \\
    D_\text{eff} &:= \lim_{t \to \infty} \frac{1}{2t} \text{Var}(X(t)) \label{eq:eff_diffusion_serfozo} \\
    &\,= \nonumber \frac{\sigma^2}{2 \mu} = \frac{1}{2 \E(\Delta T)} \Big(\text{Var}(\Delta X) + v_{\text{eff}}^2 \text{Var}(\Delta T) - 2 v_{\text{eff}} \text{Cov}(\Delta X, \Delta T)\Big). 
\end{align}
In more technically precise terms, the FCLT states:
For $r \in \mathbb{Z}_+$, define $Y_r(t) := (X(rt) - art)/(\sigma\sqrt{r/\mu}).$ If $a, \mu, \sigma, \E(M)$, and $E((\Delta T)^2)$ are all finite, then $\lim_{r \to \infty} Y_r = B$ in distribution for $t \in [0,1]$, where $\{B : t \in [0,1]\}$ is a standard Brownian motion \citep{whitt2002stochastic}.

\subsection{Notation for events within each regeneration cycle}\label{sec:notation}

The mathematical analysis in $\S$~\ref{sec:main} focuses on calculation of the moments of the cycle duration and spatial displacement (reward) in an independent cycle of the process (introduced in $\S$~\ref{sec:procedure}). Here we introduce notation for events occurring within a single regeneration cycle. We denote the number of steps in the $n$th cycle by
\begin{displaymath}
	\eta^{(n)} := \min \{k \geq 1 : J_{k+K_{n-1}+1} = 0 \},
\end{displaymath}
where $K_0 = 0$ and $K_n = \sum_{i = 1}^n \eta^{(i)}$.  
We will let $\tau_{k}^{(n)} = t_{K_{n-1}+k}-t_{K_{n-1}+k-1}$ denote the times spent in each step of  the $n$th cycle, and $ \xi_{k}^{(n)} = X(t_{K_{n-1}+k})-X(t_{K_{n-1}+k-1}) $  denote the corresponding spatial displacements.
The total time $\Delta T$ and displacement $\Delta X$ accrued in a cycle $ n\geq 2 $ before returning to the base state is then naturally the sum of these stepwise contributions:
\begin{equation} \label{eq:T-X-definition}
    \Delta T := \sum_{k=1}^{\eta^{(n)}} \tau_k^{(n)} \text{ and } \Delta X := \sum_{k=1}^{\eta^{(n)}} \xi_k^{(n)}.
\end{equation}
In what follows, we drop the superscript denoting the index $ n $ of the cycle, since the cycles have statistically independent and identically distributed behavior for $ n \geq 2$.
We will decompose each cycle into what is accrued during the first step  ($\tau_1$ and $\xi_1$) associated with the visit to the base state, and what accrues in all subsequent steps in the cycle, which we label
\begin{equation} \label{eq:tilde-T-X-definition}
    \Delta \tilde{T} := \Delta T - \tau_1 \text{ and }  \Delta \tilde X := \Delta X - \xi_1 \,.
\end{equation}

For each state $j \in S$ of the underlying Markov chain, let $\{\tau_k(j),\xi_k (j)\}_{k = 1}^\infty$ be a sequence of iid pairs of random variables drawn from the conditional joint distribution of durations and displacements occurring during a sojourn in state $ j$. The rewards collected in each step can then be written as
\begin{displaymath}
\tau_k = \sum_{j=0}^{N}\tau_k(j) 1_{\{J_k = j\}} \text{  and  } \xi_k = \sum_{j=0}^{N}\xi_k(j) 1_{\{J_k = j\}} \,.
\end{displaymath}

In the statements of our main theorems it will be useful to have a notation for a  vector of random variables with distributions for the time durations and spatial displacements that are associated with the states $S = \{0, 1, \ldots, N\}$:
\begin{equation} \label{eq:tau-xi-definition}
\vect{\tau} \defdista (\tau(0),\tau(1),\ldots,\tau(N)) \text{ and } \vect{\xi} \defdista (\xi(0),\xi(1),\ldots,\xi(N)) \,.
\end{equation}
So, for any step number $k \in \mathbb{N}$, we have that the vector $(\tau_k(0),\tau_k(1),\ldots,\tau_k(N))$ is equal in distribution to the vector $\vect{\tau}$ and likewise for the spatial displacements.

\subsection{Examples}\label{sec:examples}

The four-state example illustrated in Figure \ref{fig:trajectory} is just one of many models for intracellular transport that is carried out by multiple molecular motors. To provide context for this framework and for our result in Section~\ref{sec:main}, Proposition~\ref{thm:stoc}, here we introduce several canonical examples from the literature where intracellular transport of cargo can be modeled as a stochastic process with regenerative increments. Often, cargo fluctuations are neglected in models when a motor-cargo complex is in a processing state \citep{muller2008tug,muller2010bidirectional,kunwar2010robust}. This is equivalent to taking a limit in which the cargo is effectively instantaneously restored by the motor-cargo tether to a fixed mechanical equilibrium configuration with respect to the motor.

\vspace{2mm}
\textit{Example 1} (2-state advection-diffusion model of particle transport).\label{ex:1}
Consider a 2-state advection-diffusion model for the dynamics of protein particles (such as mRNA) as illustrated in \citet[Figure 3A]{ciocanel2017analysis}, with a freely diffusing state and an active transport state. 
Assume that the times spent by the particles in each state are drawn from an exponential distribution
\begin{align*}
\tau(0) &\sim \text{Exp}(\beta_2)\,,\\
\tau(1) &\sim \text{Exp}(\beta_1)\,.
\end{align*} 
Here $\beta_1$ and $\beta_2$ are the transition rates between states and the notation $ \text{Exp} (r) $ denotes an exponential distribution with parameter $ r $ (equal to the inverse of the mean). 
The spatial displacement in each state is given by:
\begin{align*}
\xi(0)&= \sqrt{2D\tau(0)}Z \,,\\ 
\xi(1) &= v \tau(1)\,, 
\end{align*}
where $D$ is the diffusion coefficient in the freely diffusing state, $v$ is the speed in the active transport state, and $Z$ are independent standard normal random variables.

\vspace{2mm}
\textit{Example 2} (4-state reaction-diffusion-advection model of particle transport).\label{ex:2} More realistic representations of the dynamics of cellular protein concentrations lead to considering the more complex 4-state model illustrated in \citet[Figure 3B]{ciocanel2017analysis}, where particles may diffuse, move in opposite directions, or be paused. The state durations are exponentially distributed, with the switching rates between dynamical states provided in \citet[Figure 3B]{ciocanel2017analysis}, and the spatial displacements in each state are given by:
\begin{align*}
\xi(0) &= \sqrt{2D\tau(0)}Z  \,,\\ 
\xi(1) &= v_+ \tau(1) \,,\\ 
\xi(2) &= v_- \tau(2) \,,\\ 
\xi(3) &= 0\,,
\end{align*}
with $v_+$ the particle speed in the forward active transport state and $v_-$ the particle speed in the backward active transport state.

\vspace{2mm}
\textit{Example 3} (Cooperative models of cargo transport).\label{ex:3}
Consider the cooperative transport models proposed in \citet{klumpp2005cooperative,kunwar2010robust}, where processive motors move cargo in one direction along a one-dimensional filament. These models assume a maximum number $N$ of motor proteins, firmly bound to the cargo, that may act simultaneously in pulling the cargo in a specified direction (see \citet[Figure 1]{klumpp2005cooperative} for model visualization). The biophysical state (dynamic behavior) is defined by the  number $0 \le n \le N $ of these motors that are bound to a filament and therefore actively contributing to transport.  In a state with $n$ motors attached to a filament, the cargo moves at a velocity $v_n$, motors can unbind from the filaments with rate $\epsilon_n$ or additional motors can bind to the filaments with rate $\pi_n$. The expressions for these transport model parameters are reproduced from \citet{kunwar2010robust}, together with a nonlinear force-velocity relation: 
\begin{align}
v_n(F) & = v \left( 1-\left(\frac{F}{nF_s}\right)^w \right) \,, \label{eq:cooperative_v}  \\
\epsilon_n(F) & = n \epsilon e^{F/(nF_d)}  \,, \label{eq:cooperative_eps} \\
\pi_n & = (N-n)\pi \label{eq:cooperative_pi} \,.
\end{align}
Here $v$ is the load-free velocity of the motor, $\epsilon$ is the load-free unbinding rate, and $\pi$ is the motor binding rate. $F$ is the externally applied load force, $F_s$ is the stall force and $F_d$ is the force scale of detachment. The exponent $w$ determines the nature of the force-velocity relation considered, with $w=1$ corresponding to a linear relation, $w<1$ corresponding to a concave sub-linear force-velocity curve, and $w>1$ corresponding to a convex super-linear force-velocity curve \citep{kunwar2010robust}. The times and displacements in each state $n$ (with $0 \le n \le N$ motors bound to the filaments) are therefore given by:
\begin{align}\label{eq:tau-xi-motors}
\tau(n) &\sim \text{Exp}\big(r_{\mathrm{out}}(n)\big)\,, \nonumber\\
\xi(n) &= v_n(F) \tau(n) \,, 
\end{align}
where $r_{\mathrm{out}}(n) = \epsilon_n(F)+\pi_n$ is the transition rate out of the state with $n$ motors (see \citet[Figure 1]{klumpp2005cooperative}).

\vspace{2mm}
\textit{Example 4} (Tug-of-war models of cargo transport).\label{ex:4}
Cargoes often move bidirectionally along filaments, driven by both plus and minus-directed motors. For example, kinesin moves cargo towards the plus end of microtubules while dynein moves it towards the minus end. In \citet{muller2008tug,muller2010bidirectional}, the authors propose a model where a tug-of-war between motors drives cargo in opposite directions, with transport by several motors leading to an increase in the time the cargo remains bound to a microtubule and is pulled along a particular direction. In these models, teams of maximum $N_+$ plus- and $N_-$ minus-end motors are bound to the cargo, and the biophysical state is given by the pair of indices $ (n_+,n_-)$ with $ 0 \leq n_+ \leq N $, $ 0\leq n_- \leq N $ indicating the number of plus and minus motors bound to the filament and thereby contributing actively to the transport (see \citet[Figure 1]{muller2008tug} for model visualization). A key assumption for this model is that motors interact when bound to the filament since opposing motors generate load forces, and motors moving in the same direction share the load. In addition, they assume that motors move with the same velocity as the cargo in any state \citep{muller2008tug,muller2010bidirectional}. This model uses the following expressions for the transport parameters:
\begin{align}
v_c(n_+,n_-) & = \frac{n_+F_{s+}-n_-F_{s-}}{n_+F_{s+}/v_{f+} + n_-F_{s-}/v_{b-}} \,, \label{eq:tugwar_v} \\
\epsilon_{+}(n_+) & = n_+ \epsilon_{0+} e^{F_c/(n_+F_{d+})}  \,, \label{eq:tugwar_eps}\\
\pi_+(n_+) & = (N_+-n_+)\pi_{0+} \label{eq:tugwar_pi}\,.
\end{align}
Here indices $+$ and $-$ refer to the plus- and minus-end directed motors under consideration. The model parameters are as follows: $F_s$ is the stall force, $F_d$ is the force scale for detachment, $\epsilon_0$ is the load-free unbinding rate, $\pi_0$ is the motor binding rate, $v_f$ is the forward velocity of the motor (in its preferred direction of motion), and $v_b$ is the slow backward velocity  of the motor considered. Eq.~\ref{eq:tugwar_v} applies for the case when $n_+F_{s+}>n_-F_{s-}$ (stronger plus motors, \citet{muller2008tug}), and an equivalent expression with $v_{f+}$ replaced by $v_{b+}$ and $v_{b-}$ replaced by $v_{f-}$ holds for $n_+F_{s+}\le n_-F_{s-}$ (stronger minus motors). Equivalent expressions for the binding and unbinding rates hold for the minus-end directed motors. 
In the  case of stronger plus motors, the cargo force $F_c$ 
when pulled by $n_+$ plus and $n_-$ minus motors is given by \citep{muller2008tug}:
\begin{align}
F_c(n_+,n_-) & = \lambda n_+ F_{s+} + (1-\lambda)n_- F_{s-} \,, \nonumber \\
\lambda &= \frac{1}{1+ \frac{n_+F_{s+}v_{b-}}{n_-F_{s-}v_{f+}}}\,,
\end{align}
with equivalent expressions for stronger minus motors as described above and in \citet{muller2008tug}. The times and displacements accumulated at each time step and in each state are defined as in Eq.~\ref{eq:tau-xi-motors} in Example~3.

\section{Analysis within a single cycle}
\label{sec:main}
\vspace{2mm}
From standard renewal-reward and functional central limit theorem results, which we detailed in Section \ref{sec:exs}, we have related the computation of effective velocity and diffusivity via Eqs.~\ref{eq:eff_speed_serfozo} and \ref{eq:eff_diffusion_serfozo}
to analyzing the first and second moments and correlation of the spatial displacement and time spent in each regeneration cycle.
In this section, the main result is Proposition~\ref{thm:stoc}, which provides these statistics. We begin with Lemma~\ref{lem:recursion}, by recalling a standard recursion formula for the moments of the reward accumulated until hitting a designated absorbing state. We include the proof of this lemma for completeness and as an example of the moment generating function approach we use in Lemma~\ref{lem:momentsTX}. In Proposition~\ref{thm:stoc} we address the calculation of total displacement and time duration during the regeneration cycles described in Section~\ref{sec:exs}. 

Let $0$ be the base state that marks the beginning of a new renewal cycle. We denote the set of remaining states as $S \backslash \{0\}$, and define $\tilde{P}$ as the $N \times N$ substochastic matrix containing the probabilities of transition among these non-base states only. Generally, we use the symbol \textasciitilde ~ to refer to a vector or a matrix whose components corresponding to the base state have been removed.

Let $R$ denote the total reward accumulated until the state process hits the base state. Note that the value of $R$ will depend on what the initial state of the process is. In our motor transport examples, $R$ corresponds to the time $\Delta \tilde{T}$ or the displacement $\Delta \tilde{X}$ accumulated after stepping away from the base state and before returning to the base state. Let $\rho_k$ denote the reward accumulated at each time step, recalling that time increments are denoted $\tau_k$ and displacement increments $\xi_k$ in Section~\ref{sec:notation}. 

By introducing random variables $ \rho_k (j) $ for $ j \in S$ and $ k \in \nbb $ that indicate the reward received at step $ k $ if the particle is also in state $j $ at that step, we can use indicator variables for the state to express: $\rho_k =\displaystyle\sum_{j=1}^N \rho_k(j) 1_{\{J_k = j\}}$ and
\begin{equation}
R = \sum_{k=1}^{\eta} \rho_k = \sum_{k=1}^{\eta} \sum_{j=1}^N \rho_k(j) 
    1_{\{J_k = j\}} \,.\label{eq:rdef}
\end{equation}

In the same way that we defined the  distribution for the time durations and spatial displacements through the random vectors $\vect{\tau}$ and $\vect{\xi}$ in Eq.~\ref{eq:tau-xi-definition}, we define the distribution of generic rewards through the vector of random rewards associated to each state: 
\begin{equation}\label{eq:rhodef}
    \vect{\tilde{\rho}} = (\rho (1), \rho (2), \ldots,\rho (N))\,.
\end{equation}
The tilde notation is used here to be consistent with the connotation that tilde implies the zero state is excluded.
When we need component-wise multiplication, we use the Hadamard power notation (see Eq.~\ref{eq:hadprod}):
\begin{equation}
\vect{\tilde{\rho}}^{\circ n} = (\rho^n(1), \rho^n(2),\ldots,\rho^n(N))\,.
\label{eq:hadpower}
\end{equation}

We define the moment-generating functions of the reward collected until the state process hits the base state, and of the reward in state $i$, respectively, by the following vectors:
\begin{align}
\vect{\phi}(s): \,\, \phi_i(s) &:= \E(e^{sR} \given J_1=i)\,, \text{and} \nonumber \\
\vect{\psi}(s):
\,\, \psi_i(s) &:= \E(e^{s \rho(i)})\,. \label{eq:momentgen}
\end{align}
Characteristic functions could alternatively be used to handle rewards whose higher moments are not all finite; the results for the low order moments we calculate would be the same. Note that here and in the following, we will typically use index $i$ to refer to states $i \in S \backslash \{0\}$. In Lemma~\ref{lem:recursion}, $J_1$ the state in the initial step of the process. We seek a general recursion relation for $\E(R^n|J_1 = i)$ and denote the corresponding vector of moments for all $i \in S \backslash \{0\}$ by $\E_{S \backslash \{0\}}(R^n)$. The following result is a variation on similar recursion formulas for rewards accumulated in Markov chains~\citep{hunter2008variances,palacios2009moments}.

\begin{lemma}\label{lem:recursion}
Let $\{J_k\}_{k \geq 1}$ be a time-homogeneous, positive recurrent Markov chain with a transition probability matrix $P$ (over a finite state space $S$) that has zeroes for each of its diagonal entries. Let the reward variables $R$ and $\vect{\tilde \rho}$ be defined as in Eqs.~\ref{eq:rdef} and \ref{eq:rhodef}, respectively. For $n \in \mathbb{N}$, define the column vector
\begin{equation}
    \E_{S \setminus \{0\}}(R^n) := \big(E(R^n \given J_1 = 1), \, \ldots \, , E(R^n \given J_1 = N)\big)\,.
\end{equation}
Then this vector -- the expected reward accumulated up to the first time that the state process $\{J_k\}$ hits the base state 0 -- satisfies the recursion relation
\begin{equation} \label{eq:recursion}
\begin{aligned}
\E_{S \backslash \{0\}}(R) &= (I-\tilde{P})^{-1} \E(\vect{\tilde \rho}); \\
\E_{S \backslash \{0\}}(R^n) &= (I-\tilde{P})^{-1} \left ( \E(\vect{\tilde\rho}^{\circ n}) + \sum_{m=1}^{n-1} \binom{n}{m} \mathrm{diag}(\E(\vect{\tilde\rho}^{\circ (n-m)})) \, \tilde{P} \, \E_{S \backslash \{0\}}(R^m) \right). 
\end{aligned}
\end{equation}
Here $ \tilde{P} $ is the substochastic matrix component of $ P $ excluding the base state $0$, and $\vect{\tilde{\rho}}^{\circ n}$ is the Hadamard $n$-th power vector defined in Eq.~\ref{eq:hadpower}.
\end{lemma}

\begin{proof}
Let $R$, the reward accumulated until hitting the base state $0$, be decomposed into the reward from the first and from subsequent steps as follows: $R = \rho_1 + \check{R}$. We calculate the moment-generating function of $R$ conditioned on the initial state $J_1 = i$ as follows:
\begin{align}
\phi_i(s) :=&\, \E(e^{sR}\given J_1=i) \nonumber \\
		   =&\, \sum_{j \in S} \E(e^{sR} \given J_1=i,J_2=j) P_{ij} \nonumber\\
            =&\, \sum_{j \in S} \E(e^{s\rho_1} e^{s\check{R}}\given J_1=i,J_2=j) P_{ij}\nonumber \\
            =&\, \E(e^{s \rho_1}\given J_1=i) \left(\E(e^{s\check{R}}|J_2=0)P_{i0} + \sum_{j \in S \backslash \{0\}}  \E(e^{s\check{R}} \given J_2=j) P_{ij} \right) \nonumber\\
            =&\, \E(e^{s \rho(i)}) \left(P_{i0} + \sum_{j \in S \backslash \{0\}}  \E(e^{sR} \given J_1=j) P_{ij} \right) \nonumber\\
            =&\, \psi_i(s) \left(P_{i0} + \sum_{j \in S \backslash \{0\}} \phi_j(s) P_{ij} \right) \label{eq:phi_recursion},
\end{align}
where $\psi_i(s)$ is defined in Eq.~\ref{eq:momentgen}. In the fourth line we used the Markov property, and in the fifth line we used the fact that \begin{equation*}
(\check{R} \given J_2=j) \sim (R \given J_1=j)(1-\delta_{j0})
\end{equation*}
where $ \delta_{ij}$ is the Kronecker delta function. 
Defining
\begin{align*}
f_i(s) & = \psi_i(s) P_{i0}\,, \; i \in S \backslash \{0\} \,,\\
G(s) &= \{ G(s,i,j); i,j \in S \backslash \{0\}: G(s,i,j)=\psi_i(s) P_{ij}\} \,,
\end{align*}
then we can write Eq.~\ref{eq:phi_recursion} in matrix-vector form:
\begin{align}
\vect{\phi}(s) &= \vect{f}(s) + G(s) \vect{\phi}(s)\,. \label{eq:theta_fg}
\end{align}
Since the moments of the reward before hitting the base state can be calculated using
\begin{displaymath}
\E(R^n\given J_1=i) = \frac{\partial^n}{\partial s^n} \phi_i(s)|_{s=0} \,,
\end{displaymath}
we calculate the derivatives:
\begin{align*}	
\frac{\partial^n \vect{\phi}(s)}{\partial s^n} &= \frac{\partial^n \vect{f}(s) }{\partial s^n} + \sum_{m=0}^{n}\binom{n}{m} \frac{\partial^{n-m} G(s) }{\partial s^{n-m}}  \frac{\partial^m \vect{\phi}(s) }{\partial s^m}  \,.
\end{align*}
For the first moment ($n=1$), each component yields:
\begin{align*}
\frac{\partial \phi_i(s)}{\partial s} &=  P_{i0}\E(\rho_1\given J_1=i) +   \sum_{j \in S \backslash \{0\}}  P_{ij} \E(\rho_1\given J_1=i)  + \sum_{j \in S \backslash \{0\}}  P_{ij}\frac{\partial \phi_j(s)}{\partial s} \\
&= \nonumber \E(\rho(i)) +
 \sum_{j \in S \backslash \{0\}}  P_{ij}\frac{\partial \phi_j(s)}{\partial s^n} \,.
\end{align*}
Evaluating at $s=0$ for $n=1$, we have
\begin{displaymath}
E(R \given J_1 = i) = E(\rho(i)) + \sum_{j \in S \setminus\{0\}} P_{ij} E(R \given J_n = j).
\end{displaymath}
Writing in vector form and solving for $E_{S \setminus \{0\}}(R)$ yields the first part of Eq.~\ref{eq:recursion}.

For higher-order moments ($n>1$): 
\begin{align*}
\frac{\partial^n \phi_i(s)}{\partial s^n} &=  P_{i0}\E(\rho_1^n\given J_1=i) +   \sum_{j \in S \backslash \{0\}}  P_{ij} \E(\rho_1^n\given J_1=i) \\
& \qquad + \sum_{m=1}^{n-1}\binom{n}{m} \E(\rho_1^{n-m}\given J_1=i) \sum_{j \in S \backslash \{0\}}  P_{ij} \frac{\partial^m \phi_j(s)}{\partial s^m} + \sum_{j \in S \backslash \{0\}}  P_{ij}\frac{\partial^n \phi_j(s)}{\partial s^n} \\
&= \E(\rho(i)^n) +
\nonumber     \sum_{j \in S \backslash \{0\}}  P_{ij}\frac{\partial^n \phi_j(s)}{\partial s^n} \\ 
& \qquad + \sum_{m=1}^{n-1}\binom{n}{m} \E(\rho(i)^{n-m}) \sum_{j \in S \backslash \{0\}} \!\! P_{ij} \frac{\partial^m \phi_j(s)}{\partial s^m}  \,.
\end{align*}
Evaluating at $s=0$ gives the  recursion relation expressed in the second part of 
Eq.~\ref{eq:recursion}.
\end{proof}

\begin{corollary}\label{cor:moments12}
Let $\vect{\tau}$ and $\vect{\xi}$ denote the vectors of state-dependent time duration and spatial displacements as defined in Eq.~\ref{eq:tau-xi-definition}. Let $\Delta T$ and $\Delta X$ denote the total time elapsed and displacement accumulated by a state-switching particle up until its state process $\{J_k\}_{k \geq 1}$ returns to the base state 0 (see Eqs.~\ref{eq:T-X-definition}). Moreover, recall the first-step decomposition $\Delta T = \tau_1 + \Delta \tilde T$ and $\Delta X = \xi_1 + \Delta \tilde X$ (see Eqs.~\ref{eq:tilde-T-X-definition}). Suppose that the state process $\{J_k\}_{k \geq 1}$ and its associated transition probability matrix $P$ satisfy the assumptions of Lemma \ref{lem:recursion}. 
Then
\begin{align} 
\E_{S \backslash \{0\}}(\Delta \tilde{T}) &= (I-\tilde{P})^{-1} \E(\vect{\tilde\tau}) \,, \label{eq:ETtilde}\\
\E_{S \backslash \{0\}}(\Delta \tilde{X}) &= (I-\tilde{P})^{-1} \E(\vect{\tilde\xi}) \,, \label{eq:EXtilde}\\
\E_{S \backslash \{0\}}(\Delta \tilde{T}^2) &= (I-\tilde{P})^{-1} \left(\E(\vect{\tilde\tau}^{\circ 2})  + 2 \mathrm{diag}(\E(\vect{\tilde\tau}))\tilde{P} \, \E_{S \backslash \{0\}}(\Delta \tilde{T}) \right) \, \label{eq:ETtilde2}\\
\E_{S \backslash \{0\}}(\Delta \tilde{X}^2) &= (I-\tilde{P})^{-1} \left(\E(\vect{\tilde\xi}^{\circ 2})  + 2 \mathrm{diag}(\E(\vect{\tilde\xi}))\tilde{P} \E_{S \backslash \{0\}}(\Delta \tilde{X}) \right)  \label{eq:EXtilde2}\,,
\end{align}
\end{corollary}
where $\vect{\tilde \tau}$ and $\vect{\tilde \xi}$ are the vectors of time durations and spatial displacements excluding the base state. 

\begin{proof}
These results follow directly from Lemma \ref{lem:recursion}, with $\Delta \tilde T$ and $\Delta \tilde X$ respectively playing the role of the reward $R$.
\end{proof}

\begin{lemma}\label{lem:momentsTX}
Let $\vect{\tau}$, $\vect{\xi}$, $\Delta \tilde{T}$, $\Delta \tilde{X}$, $P$, and $\{J_k\}_{k \geq 1}$ be defined as in Corollary \ref{cor:moments12}. Then
\begin{align}
\E_{S \backslash \{0\}}(\Delta \tilde{T} \Delta \tilde{X}) &= (I-\tilde{P})^{-1}\left(\E(\vect{\tilde\tau} \circ \vect{\tilde\xi})  + \mathrm{diag}(\E(\vect{\tilde\xi}))\tilde{P} \E_{S \backslash \{0\}}(\Delta \tilde{T})\right. \nonumber \\
& \qquad \qquad \qquad \qquad \quad \quad + \left. \mathrm{diag}(\E(\vect{\tilde\tau}))\tilde{P} \E_{S \backslash \{0\}}(\Delta \tilde{X})  \right) \label{eq:ETXtilde} \,.
\end{align}
\end{lemma}

\begin{proof}
We use an argument similar to the re-arrangement of the moment-generating function in Eq.~\ref{eq:theta_fg} in the proof of Lemma~\ref{lem:recursion}. Here we decompose the time and displacement into the first step after the base state and the subsequent steps: $\Delta \tilde{T} = \tau_2 + \check{T}$ and $\Delta \tilde{X} = \xi_2 + \check{X}$.
Since we are interested in the cross-moment of the duration and displacement, we consider the following moment-generating function:
\begin{align}
\phi_{i}(s,r) &= \E(e^{s \Delta \tilde{X}}e^{r \Delta \tilde{T}}|J_1 = 0, J_2 = i) \nonumber \\
		&= \sum_{j \in S} \E(e^{s \Delta \tilde{X}} e^{r \Delta \tilde{T}}|J_1=0,J_2 = i, J_3 = j) P_{ij} \nonumber \\
		&= \sum_{j \in S} \E(e^{s\xi_2}e^{r\tau_2} e^{s \check{X}}e^{r \check{T}}|J_2=i,J_3=j) P_{ij} \nonumber \\
		&= \E(e^{s\xi_2}e^{r\tau_2}|J_2=i) \sum_{j \in S}  \E(e^{s \check{X}} e^{r \check{T}}|J_2=i,J_3=j) P_{ij} \nonumber \\
            &= \E(e^{s\xi(i)}e^{r\tau(i)}) \sum_{j \in S}  \E(e^{s \check{X}} e^{r \check{T}}|J_3=j) P_{ij} \nonumber \\
            &= \psi_i(s,r) P_{i0} + \psi_i(s,r) \sum_{j \in  S \backslash \{0\}} \phi_j(s,r) P_{ij}  \label{eq:phi_s_r}\,,
\end{align}
where $\psi_i(r,s) = \E(e^{s\xi(i)}e^{r\tau(i)})$.

For the calculation of the cross-term $\E_{S \backslash \{0\}}(\Delta \tilde{T} \Delta \tilde{X})$, we note that $\frac{\partial^2 \phi_i}{\partial r \partial s}|_{s=r=0} = \E(\Delta \tilde{T} \Delta \tilde{X}|J_2=i)$ and calculate:
\begin{align}
\frac{\partial^2 \phi_i}{\partial r \partial s}&= \frac{\partial^2 }{\partial r \partial s} \left( \psi_i(r,s) \sum_{j \in S \backslash \{0\}} \phi_j(r,s) P_{ij} + \psi_i(r,s)  P_{i0} \right) \nonumber \\
		&= \frac{\partial}{\partial r} \left( \frac{\partial \psi_i(r,s)}{\partial s} \sum_{j \in S \backslash \{0\}} \phi_j(r,s) P_{ij} + \psi_i(r,s) \sum_{j \in S \backslash \{0\}} \frac{\partial \phi_j(r,s)}{\partial s} P_{ij}  + \frac{\partial \psi_i(r,s) }{\partial s} P_{i0} \right) \nonumber \\
		&= \frac{\partial^2 \psi_i(r,s)}{\partial s \partial r} \sum_{j \in S \backslash \{0\}} \phi_j(r,s) P_{ij} + \frac{\partial \psi_i(r,s)}{\partial s} \sum_{j \in S \backslash \{0\}} \frac{ \partial \phi_j(r,s)}{\partial r} P_{ij} \nonumber \\
		& \qquad + \frac{\partial \psi_i(r,s)}{\partial r} \sum_{j \in S \backslash \{0\}} \frac{\partial \phi_j(r,s)}{\partial s} P_{ij}  + \psi_i(r,s) \sum_{j \in S \backslash \{0\}} \frac{\partial^2 \phi_j(r,s)}{\partial s \partial r} P_{ij} + \frac{\partial^2 \psi_i(r,s) }{\partial s \partial r} P_{i0} \nonumber\,.
\end{align}
Evaluating the above at $s=r=0$ yields:
\begin{align}
\E(\Delta \tilde{T} \Delta \tilde{X}|J_2=i) &= \E(\tau_2 \xi_2|J_2 = i) + \E(\xi_2|J_2 = i) \sum_{j \in S \backslash \{0\}} P_{ij} E(\Delta \tilde{T}|J_2=j)  \nonumber \\
& \qquad \qquad + \E(\tau_2|J_2 = i) \sum_{j \in S \backslash \{0\}} P_{ij} \E(\Delta \tilde{X}|J_2=j) \nonumber\\
& \qquad \qquad + \sum_{j \in S \backslash \{0\}} P_{ij} \E(\Delta \tilde{T} \Delta \tilde{X}|J_2=j) \nonumber \,.
\end{align}
Therefore,
\begin{align}
\E_{S \backslash \{0\}}(\Delta \tilde{T} \Delta \tilde{X}) &= \E(\vect{\tilde\tau} \circ \vect{\tilde\xi})  + \mathrm{diag}(\E(\vect{\tilde\xi}))\tilde{P} \E_{S \backslash \{0\}}(\Delta \tilde{T}) \nonumber \\
& \qquad + \mathrm{diag}(\E(\vect{\tilde\tau}))\tilde{P} \E_{S \backslash \{0\}}(\Delta X) + \tilde{P}\E_{S \backslash \{0\}}(\Delta T \Delta X) \,, \nonumber
\end{align}
which yields equation \ref{eq:ETXtilde}.

\remark{An alternative derivation of equation~\ref{eq:ETXtilde} would be to use a polarization argument for the expectation of the product:
\begin{equation}
\E_{S \backslash \{0\}}(\Delta \tilde{T} \Delta \tilde{X})  = \frac{1}{4} \left(\E_{S \backslash \{0\}}((\Delta \tilde X + \Delta \tilde T)^2) - \E_{S \backslash \{0\}}((\Delta \tilde X- \Delta \tilde T)^2)\right)\,.
\end{equation} 
In this approach, the moment-generating function depending on both cycle time $\Delta \tilde T$ and cycle displacement $\Delta \tilde X$ introduced in Eq.~\ref{eq:phi_s_r} is not required, since Lemma~\ref{lem:recursion} can be directly applied to give explicit formulas for the second moments of the reward $R=\Delta \tilde{X} + \Delta \tilde{T}$ and $R=\Delta \tilde{X} - \Delta \tilde{T}$.
}

\end{proof}

We proceed to Proposition~\ref{thm:stoc}, which provides the quantities necessary to compute the effective velocity and diffusivity of the cargo dynamics using classical theory (see Eqs.~\ref{eq:eff_speed_serfozo} and \ref{eq:eff_diffusion_serfozo} and the procedure in \S~\ref{sec:procedure}).

\begin{proposition}[First and second order statistics of rewards in a renewal cycle]\label{thm:stoc} 

Consider a regenerative cycle of a discrete-time time-homogeneous recurrent Markov chain which takes its values in the discrete state space $S=\{0,1,2,\ldots,N\}$ with probability transition matrix $P $ with zero diagonal entries, starting at base state $ 0 $ until its first return to base state $0$. The associated time $ \Delta T$ and spatial displacement $ \Delta X $ are defined as in Eq.~\ref{eq:T-X-definition}. The random variables $\tau(0)$ and $\xi(0)$ have the distributions of the time duration and spatial displacement that are accumulated in the base state, and $\vect{p}^{(1)}$ is the vector of transition probabilities from the base state in the first step of a cycle, i.e. the first row of $P$. 
The moments of the cycle time and displacement rewards are then given by:

\begin{align} 
\E(\Delta T) &= \E(\tau(0)) + \vect{p}^{(1)} \cdot \E_{S \backslash \{0\}}(\Delta \tilde{T}) \,, \nonumber\\
\E(\Delta X) &= \E(\xi(0)) + \vect{p}^{(1)} \cdot \E_{S \backslash \{0\}}(\Delta \tilde{X}) \,, \nonumber \\
\mathrm{Var}(\Delta T) &= \mathrm{Var}(\tau(0)) +  \vect{p}^{(1)} \cdot \E_{S \backslash \{0\}}(\Delta \tilde{T}^2) - (\vect{p}^{(1)} \cdot \E_{S \backslash \{0\}}(\Delta \tilde{T}))^2\,, \nonumber \\
\mathrm{Var}(\Delta X) &= \mathrm{Var}(\xi(0)) +  \vect{p}^{(1)} \cdot \E_{S \backslash \{0\}}(\Delta \tilde{X}^2) - (\vect{p}^{(1)} \cdot \E_{S \backslash \{0\}}(\Delta \tilde{X}))^2\,, \nonumber \\
\mathrm{Cov}(\Delta X,\Delta T) &= \mathrm{Cov}(\tau(0),\xi(0)) + \vect{p}^{(1)} \cdot\E_{S \backslash \{0\}}(\Delta \tilde{T} \Delta \tilde{X}) \nonumber \\ 
& \qquad \qquad - (\vect{p}^{(1)} \cdot \E_{S \backslash \{0\}}(\Delta \tilde{T}))(\vect{p}^{(1)} \cdot \E_{S \backslash \{0\}}(\Delta \tilde{X}))  \label{eq:momts}\,,
\end{align}
where the first, second, and cross-moments of the time $\Delta \tilde T$ and the displacement $\Delta \tilde X$ are given by Eqs.~\ref{eq:ETtilde}, \ref{eq:EXtilde}, \ref{eq:ETtilde2}, \ref{eq:EXtilde2}, \ref{eq:ETXtilde} in Corollary~\ref{cor:moments12} and Lemma~\ref{lem:momentsTX}. 
\end{proposition}

\begin{proof}
With state $0$ as base state, we decompose the cycle time into the time spent in the base state $\tau_1 = \tau(0)$ and the time $\Delta \tilde{T}$ spent from leaving the base state until returning to the base state. Therefore, the total time in a cycle is given by $\Delta T = \tau_1 + \Delta \tilde{T}$, and similarly the total spatial displacement in a cycle is $\Delta X = \xi_1 + \Delta \tilde{X}$. We apply the law of total expectation by conditioning on the state $J_2$ that the process visits after the base state:
\begin{align}
\E(\Delta T) &= \E(\E(\Delta T|J_2))\nonumber \\
&= \sum_{i \in S \setminus \{0\}} \E(\Delta T|J_2 = i)P_{0i} \nonumber \\
&= \sum_{i \in S \setminus \{0\}} \E(\tau_1 + \Delta \tilde T|J_2 = i)P_{0i} \nonumber \\
&= \E (\tau(0)) + \sum_{i \in S \setminus \{0\}} \E(\Delta \tilde T|J_2 = i)P_{0i} \nonumber\\
&= \E (\tau(0)) + \vect{p}^{(1)} \cdot \E_{S \setminus \{0\}} (\Delta \tilde T) \label{eq:e0t} \,,
\end{align}
where as before $S \backslash \{0\}$ is the set of transient states and $P_{0i}$ is the probability of switching from base state $0$ to state $i$. A similar calculation applies to the first moment of the cycle reward $\E(\Delta X)$. For the second moments, we use the law of total variance as follows:
\begin{align}
\Var(\Delta T) &= \E(\Var(\Delta T|J_2)) + \Var(\E(\Delta T|J_2)) \nonumber \\
&= \E(\Var(\tau_1 + \Delta \tilde T|J_2)) + \Var(\E(\tau_1 + \Delta \tilde T|J_2)) \nonumber \\
&= \E(\Var(\tau(0))  + \Var( \Delta \tilde T|J_2)) + \Var(\E(\tau(0)) + \E(\Delta \tilde T|J_2)) \nonumber \\
&= \Var(\tau(0)) + \E(\Var( \Delta \tilde T|J_2)) + \Var(\E(\Delta \tilde T|J_2)) \nonumber \\
&= \Var(\tau(0)) + \sum_{i \in S\setminus \{0\}} \Var( \Delta \tilde T|J_2=i) P_{0i} \nonumber \\
&\qquad \qquad + \sum_{i \in S\setminus \{0\}} (\E(\Delta \tilde T|J_2=i))^2 P_{0i} - \left(\sum_{i \in S\setminus \{0\}} \E(\Delta \tilde T|J_2=i)P_{0i} \right)^2 \nonumber \\
&= \Var(\tau(0)) + \sum_{i \in S\setminus \{0\}} \E( \Delta \tilde T^2|J_2=i) P_{0i} - \left(\sum_{i \in S\setminus \{0\}} \E(\Delta \tilde T|J_2=i)P_{0i} \right)^2 \nonumber \\
&= \Var(\tau(0)) + \vect{p}^{(1)} \cdot \E_{S \setminus \{0\}}(\Delta \tilde T^2) -(\vect{p}^{(1)} \cdot \E_{S \setminus \{0\}}(\Delta \tilde T))^2
\label{eq:e0t2} \,,
\end{align}
and similarly for $\Var(\Delta X)$. The covariance term can then be obtained via the polarization formula from the formulas for the variances.

\end{proof}

\section{Application to models of intracellular transport}
\label{sec:application_intra}

Proposition~\ref{thm:stoc} and the calculation procedure in \S~\ref{sec:procedure} can be applied to understand the long-term dynamics of protein intracellular transport described in \S~\ref{sec:exs} in Examples 1 and 2. The effective velocity and diffusivity of proteins are key in understanding large timescale processes such as mRNA localization in frog oocytes \citep{ciocanel2017analysis} and cell polarization in the budding yeast \citep{bressloff2015stochastic}.

\subsection{2-state advection-diffusion model of particle transport}

In the following, we consider the 2-state transport-diffusion model for the dynamics of mRNA particles described in Example 1 and illustrated in \citet[Figure 3A]{ciocanel2017analysis}. We show how the calculations in Proposition~\ref{thm:stoc} can be applied to determine the large-time effective velocity and diffusivity of the particles. 

In the 2-state model, the probability transition matrix is simply $P = \begin{pmatrix}
0 & 1\\
1 & 0
\end{pmatrix}\,.$ We take the diffusing state as the base state. The substochastic matrix of the probabilities of transition between the other states~\citep{bhat2002elements} is then simply the scalar $\tilde{P}  = 0$ in this case, while the vector of transition out of the base state is simply $ \vect{p}^{(1)} = [1].$  

The first and second moments of the cycle duration are given by Eqs~\ref{eq:ETtilde} and \ref{eq:ETtilde2} with $(I-\tilde{P} )^{-1} = 1$. Similarly, the moments of the spatial displacement are given by Eqs.~\ref{eq:EXtilde} and \ref{eq:EXtilde2}. In this model, we have that $S\backslash \{0\}=\{1\}$ and $\tilde\tau_k^{\circ n} = \tau_k^n(1)$ for the time reward and $\xi_k^{\circ n} = \xi_k^n(1)$ for the spatial displacement reward  in the active transport state.
In the 2-state system, these values are simply scalars:
\begin{displaymath}
\begin{aligned}
\E_1(\Delta \tilde{T})  &= \E(\tau(1)) = 1/\beta_1 \,, \\
\E_1(\Delta \tilde{X}) &= \E(\xi(1)) = v/\beta_1  \,,\\
\E_1(\Delta \tilde{T} \Delta \tilde{X}) &= \E(\tau(1) \xi(1)) = 2v/\beta_1^2 \,,\\
\E_1(\Delta \tilde{T}^2) &= \E(\tau(1)^2) = 2/\beta_1^2\,,\\
\E_1(\Delta \tilde{X}^2) &= \E(\xi(1)^2) = 2v^2/\beta_1^2 \,.
\end{aligned} 
\end{displaymath}
The statistics of the cycle are therefore given by:
\begin{displaymath}
\begin{aligned}
\E(\Delta T) &=\E(\tau(0)) + \E_1(\Delta \tilde{T}) \vect{p}^{(1)}(1) = \frac{1}{\beta_2} +  \frac{1}{\beta_1} \,, \label{eq:et}\\
\E(\Delta X) &=\E(\xi(0)) + \E_1(\Delta \tilde{X}) \vect{p}^{(1)}(1) = 0 + v/\beta_1 = \frac{v}{\beta_1}  \,,\\
\text{Var}(\Delta T) &= \text{Var}(\tau(0)) + \E_1(\Delta \tilde{T}^2) \vect{p}^{(1)}(1) - (\E_1(\Delta \tilde{T}) \vect{p}^{(1)}(1))^2 = \frac{1}{\beta_2^2} + \frac{1}{\beta_1^2}\,,\\
\text{Var}(\Delta X) &= \text{Var}(\xi(0)) + \E_1(\Delta \tilde{X}^2) \vect{p}^{(1)}(1) - (\E_1(\Delta \tilde{X}) \vect{p}^{(1)}(1))^2 = \frac{2D}{\beta_2} + \frac{v^2}{\beta_1^2}  \,,\\
\text{Cov}(\Delta T,\Delta X) &= \text{Cov}(\tau(0),\xi(0)) + \E_1(\Delta \tilde{T} \Delta \tilde{X}) \vect{p}^{(1)}(1) \\
& \qquad \qquad - \left(\E_1(\Delta \tilde{T}) \vect{p}^{(1)}(1)\right) \left(\E_1(\Delta \tilde{X}) \vect{p}^{(1)}(1)\right)= \frac{v}{\beta_1^2}\,.
\end{aligned}
\end{displaymath}
Eqs.~\ref{eq:eff_speed_serfozo} and \ref{eq:eff_diffusion_serfozo} then provide expressions for the effective velocity and diffusivity of the particles as in \citet{hughes2011matrix,hughes2012kinesins,whitt2002stochastic}:
\begin{align}
v_{\mathrm{eff}} & = \frac{\E(\Delta X)}{\E(\Delta T)} = v \frac{\beta_2}{\beta_1+\beta_2}\,, \nonumber \\
D_{\mathrm{eff}} & = \frac{1}{2 \E(\Delta T)} (v_{\mathrm{eff}}^2 \text{Var}(\Delta T) + \text{Var}(\Delta X) - 2v_{\mathrm{eff}}\text{cov}(\Delta T,\Delta X)) \nonumber \\
&= D\frac{\beta_1}{\beta_1+\beta_2} + v^2 \frac{\beta_1 \beta_2}{(\beta_1+\beta_2)^3} \nonumber \,.
\end{align}

Note that the effective velocity is given by the speed in the transport state multiplied by the fraction of time the mRNA particles spend in the moving state. The effective diffusivity has a more complicated expression, but clearly shows the dependence of this quantity on each model parameter. These expressions agree with the results of Eqs.~\ref{eq:ds_vel}, \ref{eq:ds_diff} as outlined in \citet{ciocanel2017}.

\subsection{4-state advection-reaction-diffusion model of particle transport}
Our calculation procedure and Proposition~\ref{thm:stoc} extend to more complicated and realistic models such as the 4-state model described in Example 2 and illustrated in \citet[Figure 3B]{ciocanel2017analysis}. By considering the stochastic transitions between dynamic states and the durations and displacements accumulated in each state, the effective velocity and diffusion of cargo can be calculated in an intuitive way even for such complex models with many transition states. Since this approach requires calculating the inverse of the invertible matrix $I-\tilde{P}$ (see \citet{bhat2002elements,dobrow2016introduction}) to determine the fundamental matrix, the approach presented here is easily implemented in a software package such as \textit{Mathematica} or \textit{Matlab} for symbolic derivation of the effective transport properties for models with multiple states (see sample code in the repository on \cite{Github_effdiff2019}). 

In Figure~\ref{fig:comparison_4state}, we illustrate the good agreement of the results in \citet{ciocanel2017analysis} with our calculation procedure in \S~\ref{sec:procedure} (Proposition~\ref{thm:stoc} combined with Eqs.~\ref{eq:eff_speed_serfozo} and \ref{eq:eff_diffusion_serfozo}) based on 15 sets of parameters estimated in \citet{ciocanel2017analysis}. In addition, we validate results from both approaches by carrying out numerical simulations of the particle transport process and empirically estimating the effective transport properties. In particular, we set up a Markov chain of the 4-state model in \citet[Figure 3B]{ciocanel2017analysis}. For each parameter set, we consider $N_R = 500$ stochastic realizations of the dynamics and for each iteration, we run the process until a fixed large time $T_f = 5\times10^4$, which also keeps the computation feasible. We then estimate the effective velocity and diffusivity as follows:
\begin{align}
v_{\mathrm{eff}} &\approx \frac{(\sum_{i=1}^{N_R} X_i(T_f))/N_R} {T_f}\,,\\
D_{\mathrm{eff}} &\approx \frac{\left(\sum_{i=1}^{N_R}(X_i(T_f) - (\sum_{i=1}^{N_R}X_i(T_f))/N_R)^2\right)/(N_R-1)}{2T_f}\,,
\end{align}
where $X_i(T_f)$ are the simulated final positions of the particle at time $T_f$ in iteration $i$. 

The different parameter sets (labeled by index) in Figure~\ref{fig:comparison_4state} correspond to simulations using parameter estimates based on FRAP mRNA data from different frog oocytes in \citet{gagnon2013directional,ciocanel2017analysis}. The good agreement of the theoretical and simulated effective velocity and diffusivity shows that the analytical approach proposed is a good alternative to potentially costly simulations of the stochastic process up to a large time. 

\begin{figure}[t]
\centering
  \includegraphics[height=45mm]{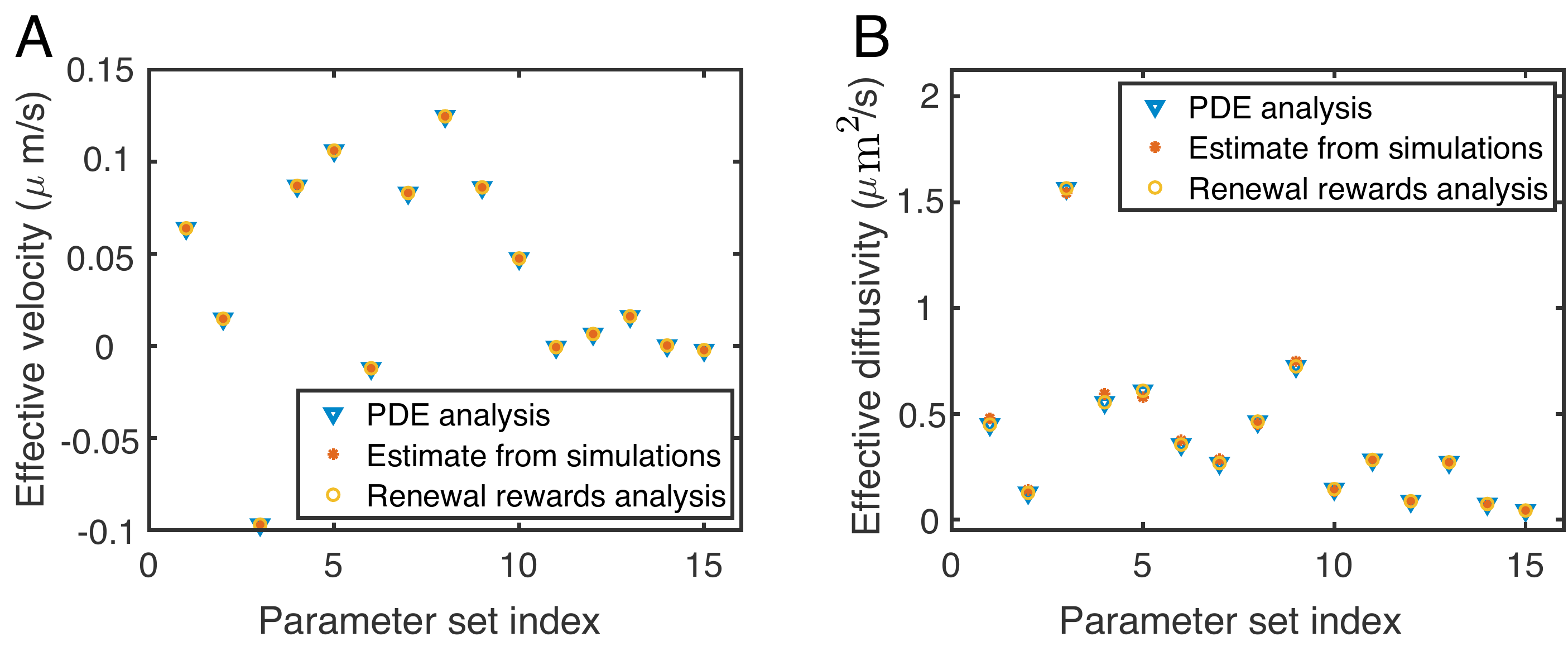}
\caption{Effective velocity (A) and effective diffusivity (B) of particles switching between diffusion, bidirectional transport, and stationary states as in \citet[Figure 3B]{ciocanel2017analysis} for different parameter sets. Blue triangles correspond to predictions based on the homogenization or equivalent analysis~\citep{ciocanel2017analysis} of the corresponding PDEs (Eq.~\ref{eq:n-state-mod}), filled red dots correspond to estimates from multiple simulated realizations of the Markov chain, and yellow circles correspond to predictions based on analysis of the corresponding renewal process model combined with Proposition~\ref{thm:stoc}. }\label{fig:comparison_4state}
\end{figure}

\section{Application to cooperative and tug-of-war models of cargo transport}
\label{sec:application_tugwar}

The framework presented here also extends to models of cargo particles driven by changing numbers of motor proteins. The analytical calculation of transport properties of cargo pulled by motors in the same or opposite directions could replace or complement costly numerical simulations of individual cargo trajectories. In the following, we consider both models of cooperative cargo transport with identical motors \citep{klumpp2005cooperative,kunwar2010robust} and tug-of-war models of bidirectional transport driven by identical or different motors moving in opposite directions \citep{muller2008tug,muller2010bidirectional}.

\subsection{Cooperative models of cargo transport}\label{sec:cooperative}
We start by considering the cooperative transport models described in \S~\ref{sec:exs}, Example~3, and studied in \citet{klumpp2005cooperative,kunwar2010robust}, with processive motors that move along a one-dimensional microtubule and transport cargo in only one direction. The dynamics is described by the force-driven velocities $v_n$, unbinding rates $\epsilon_n$ and binding rates $\pi_n$ in each state with $n$ motors bound to the microtubule and moving the cargo (see Eqs.~\ref{eq:cooperative_v}, \ref{eq:cooperative_eps}, and \ref{eq:cooperative_pi}). In this section, we use the kinetic parameters for conventional kinesin-1 provided in \citet{klumpp2005cooperative}.

Our calculation of the effective velocity of cargo agrees with the derivation in \citet{klumpp2005cooperative}, which uses the stationary solution of the master equation for probabilities of the cargo being in each state (i.e. carried by $n$ motors). We note that there are two notions of effective velocity (and diffusivity) that can be used in studying this model: one is to calculate the effective velocity of the cargo averaged over the bound states only (the asymptotic velocity without detachment along a theoretical infinite length microtubule) \citep{klumpp2005cooperative,kunwar2010robust}, and the second is to calculate the overall effective velocity that also accounts for periods of detachment from microtubules. For the $N=2$ motors model, \citet{klumpp2005cooperative} and \citet{kunwar2010robust} report the average velocity for bound cargo (first notion): 
\begin{equation}
\label{eq:veff_normbound}
v_{\mathrm{eff}} = v_1 \frac{\pi_0 \epsilon_2}{\pi_0 \epsilon_2+\pi_0\pi_1} + v_2 \frac{\pi_0 \pi_1}{\pi_0 \epsilon_2+\pi_0\pi_1}\,.
\end{equation}
Since we are interested in the overall effective velocity of the particles in the context of their full dynamics, we include the state where no motors are bound to the filament in our calculation, so that the effective velocity with respect to the overall dynamics is given by:
\begin{align}\label{eq:veff_normunbound}
v_{\mathrm{eff}} & = v_1 \frac{\pi_0 \epsilon_2}{\epsilon_1\epsilon_2 + \pi_0 \epsilon_2+\pi_0\pi_1} + v_2 \frac{\pi_0 \pi_1}{\epsilon_1\epsilon_2 + \pi_0 \epsilon_2+\pi_0\pi_1}\,.
\end{align}

Using the calculation of the overall effective velocity in \eqref{eq:veff_normunbound}, we predict a similar dependence  of the effective velocity under a range of force loads as in~\citet{klumpp2005cooperative} using the formula~\eqref{eq:veff_normbound}. The dashed curves in Figure~\ref{fig:subsuperlinear}A,C agree with the behavior of sub- and super-linear motors under different load forces as reported in \citet[Figure 2C-D]{kunwar2010robust}, including the fact that  sub-linear motors have lower effective velocities for any choice of the load force and for all maximum motor numbers $N$ considered (A, $w=0.5$), while super-linear motors are faster and therefore have larger effective velocities than linear motors (C, $w=2$).  

\begin{figure}[t]
\centering
\includegraphics[height=92mm]{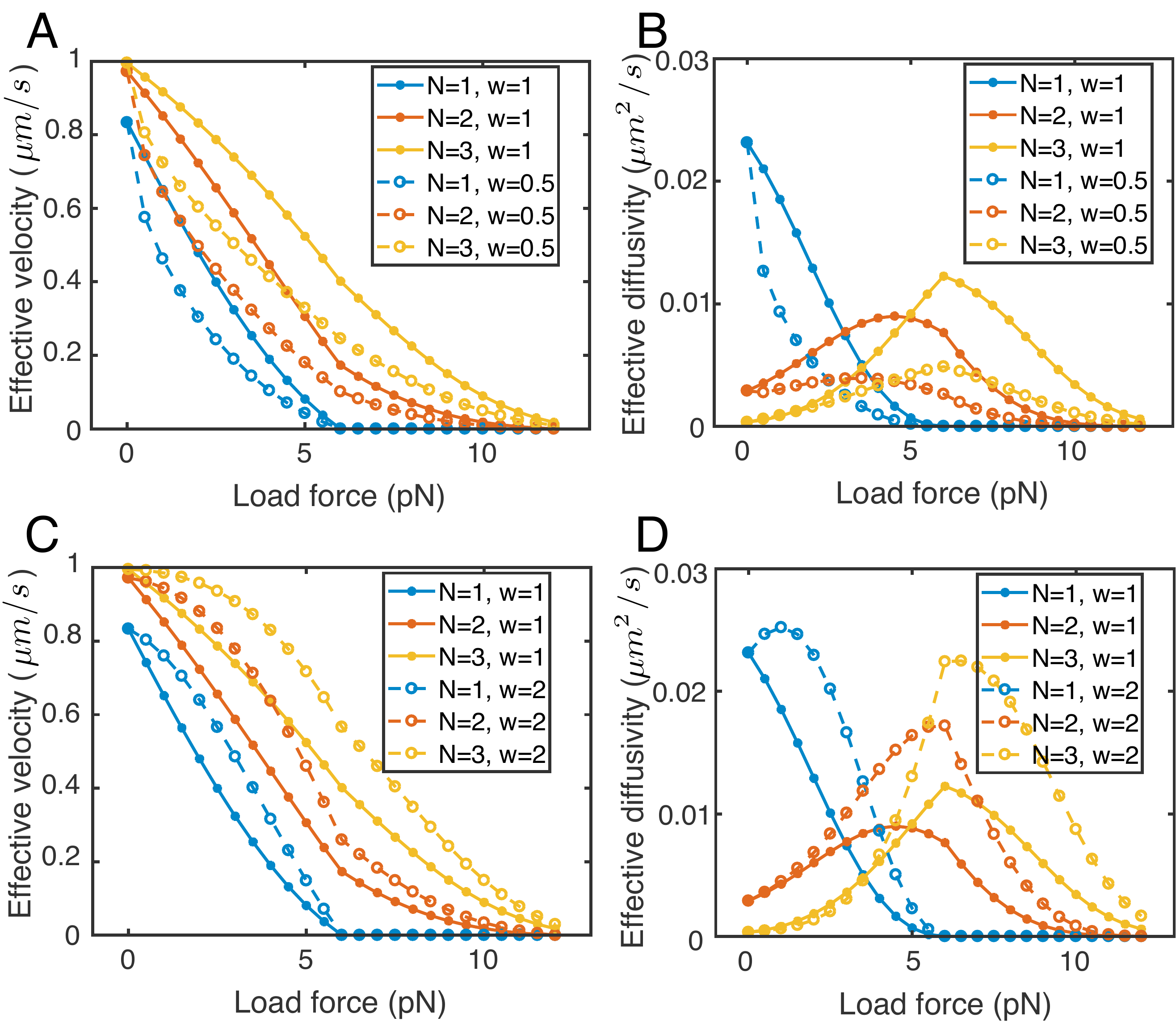} 
\caption{Effective velocity (A,C) and effective diffusivity (B,D) of cargo driven by a maximal number $N$ of forward motor proteins as a function of the load force under various velocity-force exponents $ w $ (Eq.~\ref{eq:cooperative_v}). The stall force used for kinesin is $F_s = 6$~pN. Solid lines correspond to motors with a linear force-velocity relation and dashed lines correspond to sub-linear motors with a convex-up force-velocity relation (top row), and respectively to super-linear motors with a concave force-velocity relation (bottom row).}\label{fig:subsuperlinear}
\end{figure}

The insight from our method lies in the prediction of the effective diffusivity as a function of load for each type of motor. Figure~\ref{fig:subsuperlinear}B,D show that the $N=1$ motor transport case has a large effective diffusivity under no load because of the switching between the paused and moving states. As the force load increases to stall $F_s$, the velocity of the single motor state decreases to 0: $v_1(F) = v\left(1-F/F_s \right)$. Therefore, the active transport state switches to a stationary state at $F=F_s=6$~pN, leading to decreased effective diffusivity as the cargo switches between dynamic states with similar behaviors.

For $N=2$ and $N=3$, the calculation of the effective diffusivity allows us to re-visit the cooperative transport models for a large range of load forces and observe a new phenomenon in the classical models of \citet{klumpp2005cooperative,kunwar2010robust}. The broader sweep of the load force parameter in Figure~\ref{fig:subsuperlinear}B,D shows a non-monotonic dependence of the effective diffusivity on load force for all types of motors considered (linear, sublinear, and superlinear), with an increase in effective diffusivity of cargo at low load forces and a decrease at large load forces. While it is not immediately clear what leads to this phenomenon, we conjecture that this observation may be a result of the balance between two competing effects: on the one hand, as the load increases, there is more detachment of motors (see \eqref{eq:cooperative_eps}) and thus more frequent switches between transport and stationary states, leading to an increase in effective diffusivity; on the other hand, the increase in load force leads to a decrease in the speeds of the motor-driven cargo states (see \eqref{eq:cooperative_v}), and thus a decrease in effective diffusivity.

\subsection{Tug-of-war models of cargo transport} \label{sec:tugwar}
In Example 4 in \S~\ref{sec:exs}, we consider the case where plus- and minus-directed motors can drive cargo bidirectionally along filaments. The cargo velocities $v_c(n_+,n_-)$, unbinding rates $\epsilon_{+/-}(n_+,n_-)$ and binding rates $\pi_{+/-}(n_{+/-})$ depend on the number of plus motors $n_+$ and minus motors $n_-$ at each state.
\begin{figure}[t]
\centering
\includegraphics[height=92mm]{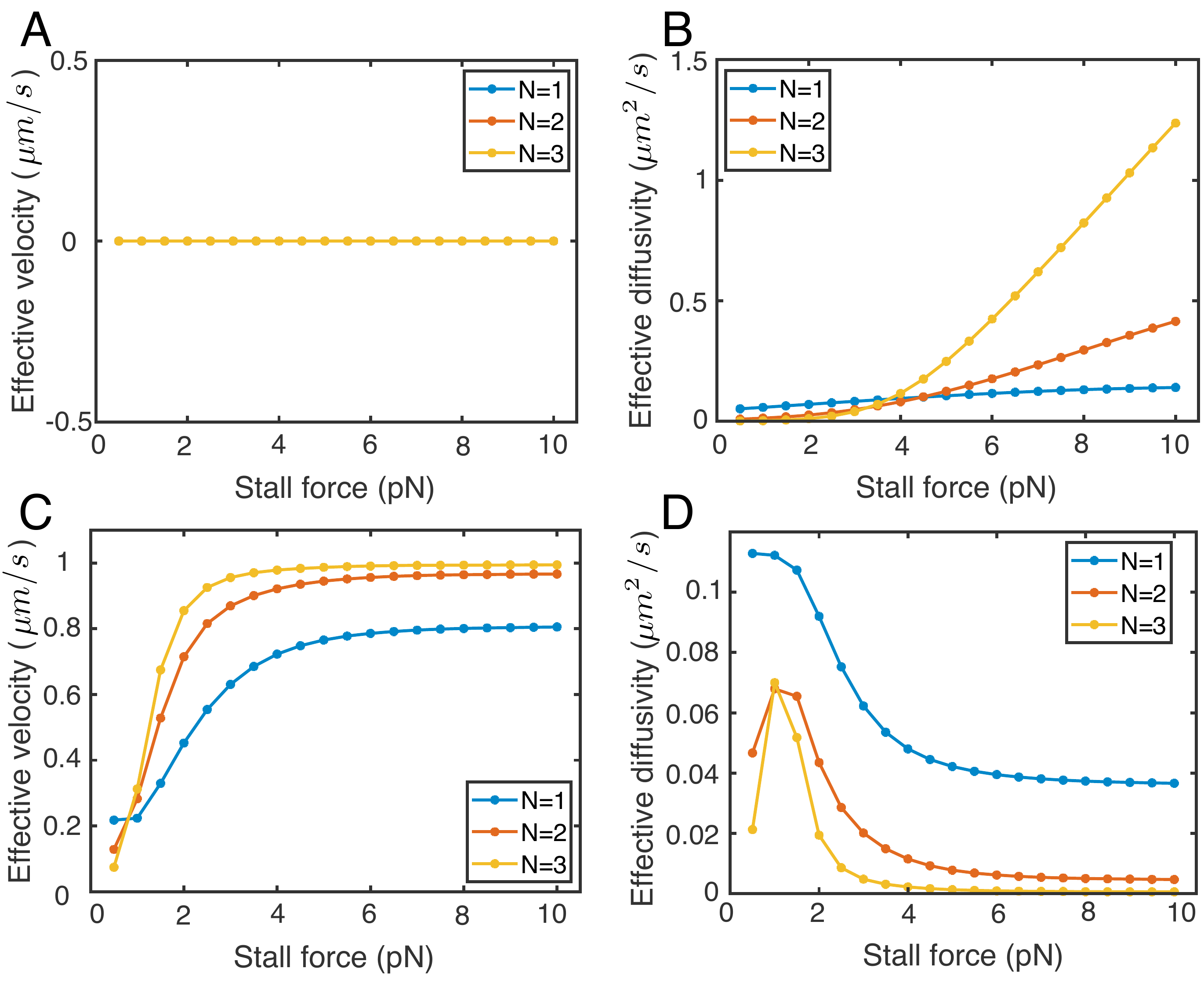} 
\caption{Effective velocity (A,C) and effective diffusivity (B,D) of cargo driven by maximum $N$ forward and maximum $N$ backward motor proteins as a function of kinesin-1 stall force $F_s$; the detachment force is $F_d = 3$ pN. Panels (A,B) correspond to identical forward and backward motors with kinetic parameters for kinesin-1, and panels (C,D) correspond to kinesin-1 forward motors and conventional dynein backward motors.}\label{fig:tugwar_kin}
\end{figure}

\textit{Identical plus and minus motors.} 
With kinesin parameters drawn from \citet[Table 1]{muller2008tug}, we first calculate the transport properties of cargo in these models for identical plus and minus motors in equal numbers ($N_+=N_-$). We vary the stall force of the kinesin motor to determine if the theoretical effective velocity and diffusivity capture the differences obtained in the numerical simulation studies in~\citet{muller2008tug,muller2010bidirectional} for weak motors (small stall to detachment force ratio $f = F_s/F_d$) and strong motors (large $f$). As expected, the effective velocity in this symmetric case of identical motors is zero for all stall forces (see Figure~\ref{fig:tugwar_kin}A). The predicted effective diffusivity in Figure~\ref{fig:tugwar_kin}B shows that for weak motors, the effective diffusivity is small, and different maximum numbers of motors do not lead to significant differences. This is similar to the results in \citet{muller2008tug}, where the simulated cargo trajectories show small fluctuations and the probability distribution for the velocity has a single maximum peak corresponding to approximately equal numbers of plus and minus motors attached. However, for strong motors with a larger stall to detachment force ratio, the effective diffusivity increases considerably for all models. This is consistent with the observation in \citet{muller2008tug} that strong motors lead to cascades of unbinding of minus motors until only plus motors stay bound (and vice versa), so that the spread of the cargo position is predicted to be larger. The larger motor numbers lead to a more significant increase in effective diffusivity as observed in \citet{muller2010bidirectional}, where the simulated diffusion coefficient grows exponentially with motor numbers and therefore leads to a more productive search of target destinations throughout the domain \citep{muller2010bidirectional}.

It is worth noting that the method we develop in \S~\ref{sec:procedure} extends to cases where slow diffusive transport rather than pausing is observed in the unbound state (see \S~\ref{sec:application_intra} for another example with a diffusive state). As expected, when the cargo has an intrinsic diffusion coefficient that is non-zero, the effective velocity of the cargo does not change, however the effective diffusivity is consistently larger than in the case where the unbound cargo is fully stationary (results not shown).

\textit{Distinct plus and minus motors.} 
When considering dynein as the minus-end directed motor in the bidirectional transport model, we use the kinetic parameters estimated to fit \textit{Drosophila} lipid droplet transport in \citet[Table 1]{muller2008tug}. Figure~\ref{fig:tugwar_kin}C shows that the cargo is predicted to move in the forward (kinesin-driven) direction with a positive effective velocity. We again observe increased transport efficiency for larger numbers of motors. With increasing stall force, the velocity of individual runs in each state increases and therefore the effective velocity increases and then plateaus. This asymmetric motor case also results in effective diffusivity that decreases past a small stall force and then stabilizes (see Figure~\ref{fig:tugwar_kin}D). Since the kinesin motor dominates the dynamics, there are fewer excursions backwards than in the case of identical motors, so that the effective diffusivity is an order of magnitude smaller. Larger teams of motors regularize the dynamics and display decreased effective diffusivity.

\subsection{Reattachment in models of cargo transport}
\begin{figure}[t]
\centering
\includegraphics[height = 92mm]{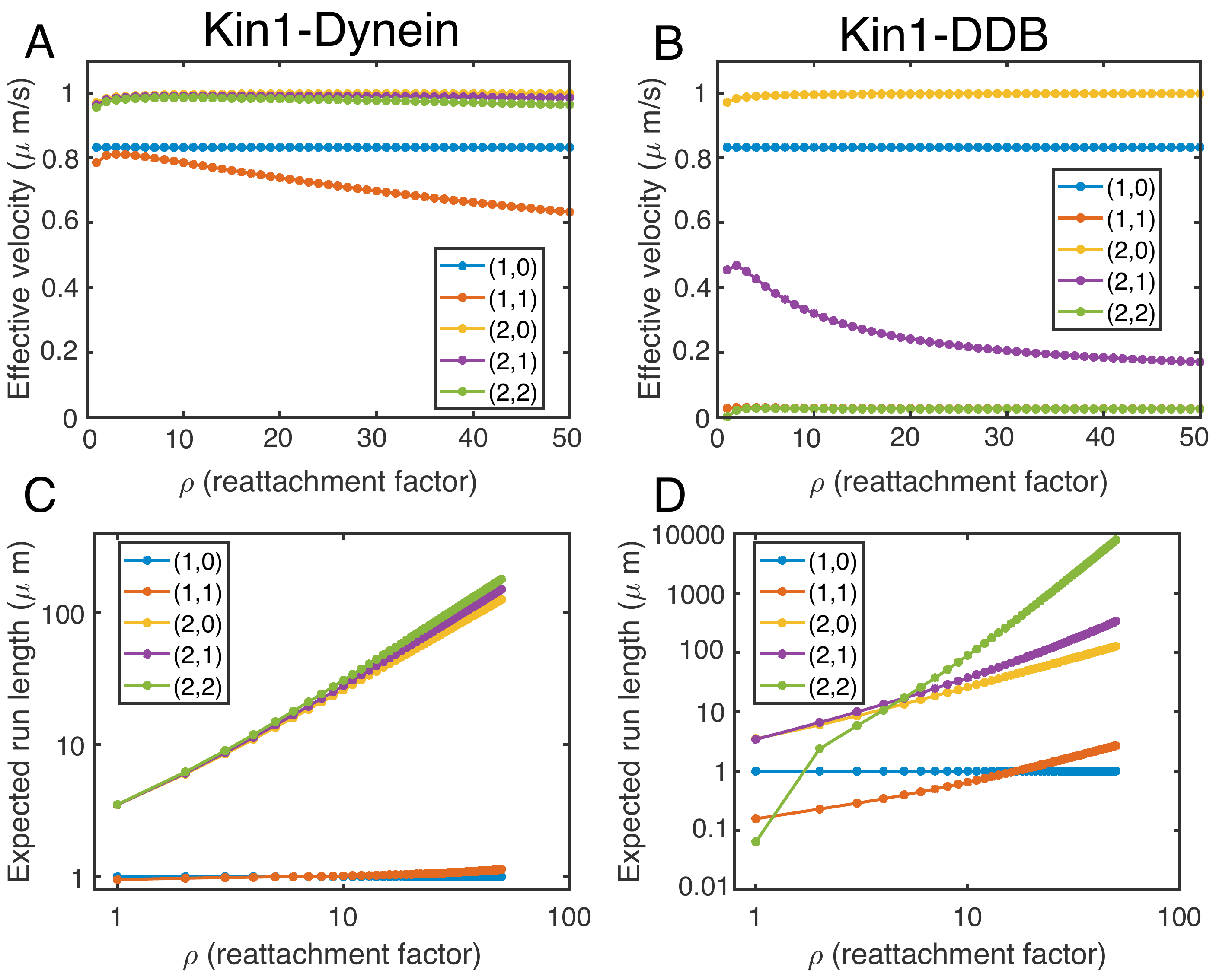} 
\caption{Effective velocity (A,B) and expected run length (C,D) of cargo driven by maximum $N_1$ forward (kinesin-1) and maximum $N_2$ backward (dynein) motor proteins $(N_1,N_2)$ as a function of the reattachment factor $\rho$. Panels (A,C) use conventional dynein motor kinetics as in \citet{muller2008tug,muller2010bidirectional} while panels (B,D) use dynein-dynactin-BicD2 (DDB) complex parameters as in \citet{ohashi2019load}. The run lengths in panels (B) and (D) are plotted on a log-log scale to allow for visualization of differences between the models considered. The definitions of effective velocity and run length used are provided in the text.}\label{fig:tugwar_rho}
\end{figure}

In vitro experiments have suggested that binding rates of molecular motors at specific locations may be regulated by the concentration of the same or opposite-directed motors \citep{hancock2014bidirectional}, as well as by the availability of microtubule filaments. To test for the impact of reattachment kinetics in the standard transport models of \citet{muller2008tug,muller2010bidirectional}, we modify the binding rate in \eqref{eq:tugwar_pi} to account for a higher likelihood of reattachment when a motor (of either type) is already attached to the microtubule:
\begin{align}
\pi_+(n_+,n_-) & = \begin{cases}
N_+\pi_{0+}, \mathrm{\: \: if \: \: n_++n_-=0} \,,\\
(N_+-n_+)\rho\pi_{0+}, \mathrm{\: \: \: \: else}  \label{eq:tugwar_pi_reattach}\,.
\end{cases}
\end{align}
Here $\rho>0$ denotes the reattachment factor, and an equivalent expression is valid for the binding rate for minus motors $\pi_-(n_+,n_-) $. $\rho=1$ corresponds to the binding kinetics in the previous sections, and $\rho>1$ denotes an increased reattachment likelihood when other motors are attached.

Figure~\ref{fig:tugwar_rho} illustrates the effective velocity (panels A,B) and the expected cargo run length (Eq.~\ref{eq:momts} for $\Delta X$, panels C,D) for values of $\rho$ ranging from $1$ to $50$, in the context of models labeled $(N_1,N_2)$ with transport driven by maximum $N_1$ forward (kinesin-1) motors and $N_2$ backward (dynein) motors. Here we report the overall effective velocity of the cargo according to the second definition in \S~\ref{sec:cooperative}, which includes both attached and detached cargo states in the calculation. In addition, the mean run length is calculated as the mean total displacement over a cycle starting with all motors detached until its return to a completely detached state, namely $\E(\Delta X)$ in Eqs.~\ref{eq:momts}. Note the base state of complete detachment makes no contribution to the mean displacement. 

Classically in tug-of-war modeling, dynein has been viewed as a ``weaker partner'' than kinesin-family motors. In this parameter regime \citep{muller2008tug,muller2010bidirectional}, dynein has both a smaller stall force and smaller critical detachment force than kinesin-1. As a result, when equal numbers of kinesin-1 and dynein are simultaneously attached, kinesin-1 dominates transport.  However, it has recently been shown that it might not be realistic to consider dynein in the absence of its helper proteins, particularly dynactin and BicD2. Together, these form a complex referred to as DDB, and the associated parameter values \citep{ohashi2019load} are much more ``competitive'' with kinesin-1 in a tug-of-war scenario. 

In Figure~\ref{fig:tugwar_rho}, we display the effective velocity and expected run length of kinesin-1 vs dynein (panels A,C), and kinesin-1 versus DDB (panels B,D) dynamics. In Figure~\ref{fig:tugwar_rho}A, the effective velocity of cargo driven by teams of motors approaches the effective speed predicted for kinesin-only motor teams (models $(1,0)$ and $(2,0)$) for small values of $\rho$, but then decreases as $\rho$ becomes larger for conventional dynein motility. As observed in recent studies, activated dynein competes more efficiently with kinesin and therefore the teams of opposite-directed motors are consistently slower than teams consisting of only the forward kinesin motor protein in Figure~\ref{fig:tugwar_rho}B \citep{ohashi2019load}. The expected run lengths in Figure~\ref{fig:tugwar_rho}C,D illustrate that teams of multiple motors are characterized by significantly increased processivity on microtubules as the reattachment factor becomes larger. When considering conventional dynein, the difference in processive cargo motion between the cooperative and tug-of-war models is only observed at large values of the reattachment constant $\rho$ ($>10$, see Figure~\ref{fig:tugwar_rho}C). This is due to the fact that the backward motor (conventional dynein) in the \citet{muller2008tug} model is weak with a small detachment force, so that overcoming the large dynein unbinding rate requires large values of the reattachment factor. On the other hand, activated dynein in the DDB complex is a more equal competitor to kinesin, with predictions of the expected run length in Figure~\ref{fig:tugwar_rho}D confirming the experimental observations of larger unloaded run lengths in \citet{ohashi2019load}.

\subsection{Microtubule sliding model}
\begin{figure}[t]
\centering
\includegraphics[height = 45mm]{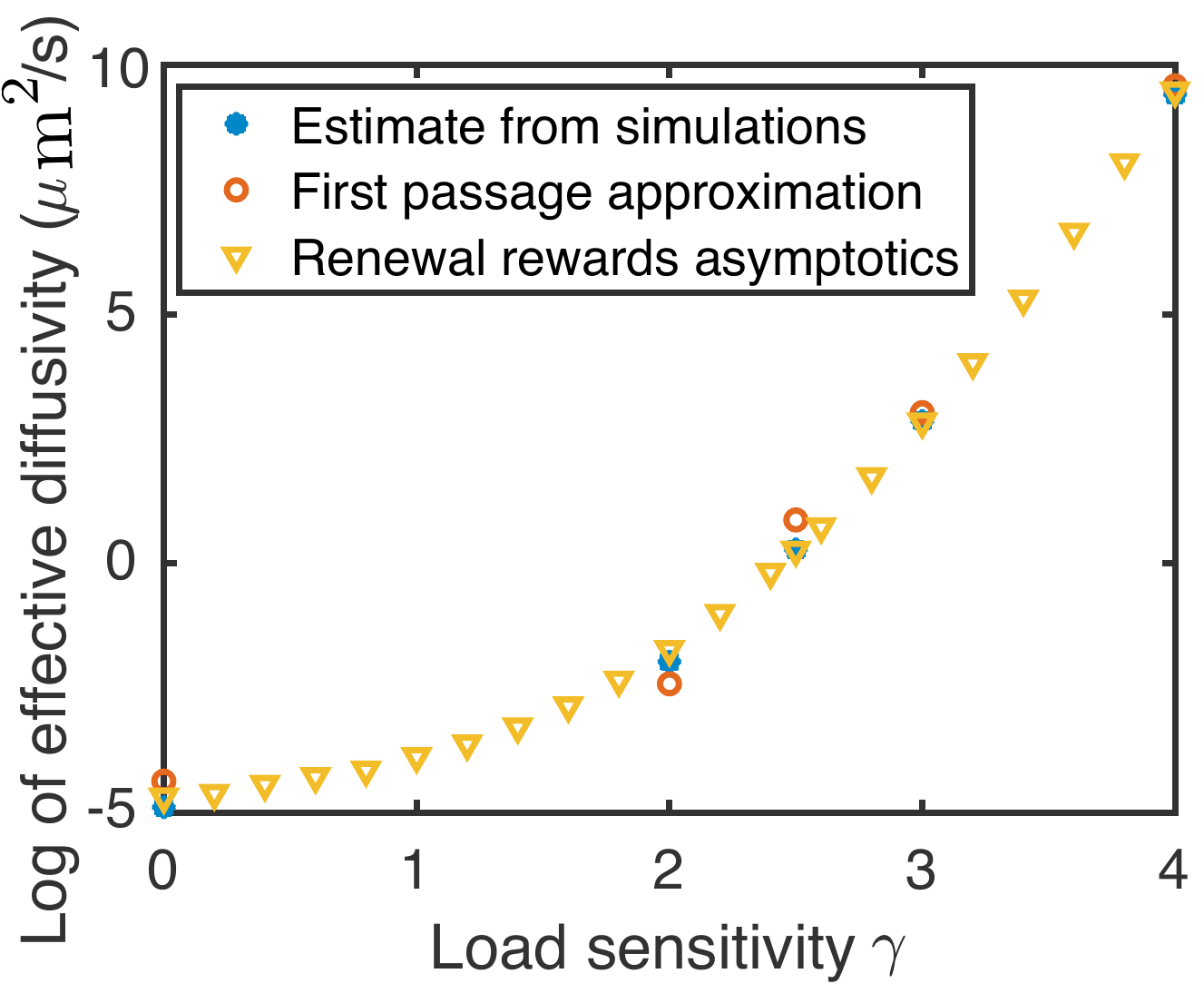} 
\caption{Comparison of effective diffusivity estimates for parallel microtubules driven bidirectionally by kinesin motors as a function of the scaled load sensitivity $\gamma$: stochastic simulations from \citet{allard2019sliding} marked with blue stars, first passage time approximation in \citet{allard2019sliding} marked with red circles, and renewal rewards calculation marked with yellow triangles. Following \citet{allard2019sliding}, we only allow states with $i$ kinesin motors moving forward and $K-i$ kinesin motors moving backward, with $0\le i \le K$ and $K=35$. The vertical axis is plotted on a log scale to allow for visualization of differences between effective diffusivity estimated from simulations and analytical approximations.}\label{fig:mt_sliding}
\end{figure}

As a final example of the applicability of our method, we consider a recent investigation into microtubule motility and sliding by \citet{allard2019sliding}. The authors consider a continuous-time Markov chain model of the interaction of two parallel microtubules, cross-linked and moved by multiple identical kinesin motors. Depending on which microtubule the motor heads are attached to, they push the microtubule pair apart in one of two directions (one of which is arbitrarily assigned to be ``positive''). The model assumes that motor attachment to microtubules occurs quickly relative to detachment, allowing a reduced number of dynamic states with microtubules driven by $i$ motors pushing in the positive direction  and $K-i$ motors pushing in the negative direction (where $K$ is the maximal number of motors that fit the overlap region between the two parallel microtubules). The detachment rates are therefore given by 
\begin{align}\label{eq:ki_rates_sliding}
\kappa_i^+ &= (K-i)\kappa_0 \exp{\left(\gamma \frac{i}{K}\right)} \,, \nonumber\\
\kappa_i^- &= i\kappa_0 \exp{\left(\gamma \frac{K-i}{K}\right)} \,, 
\end{align}
where $\kappa_i^+$ is the rate at which a motor pulling in the negative direction is replaced by one pulling in the positive direction, $\kappa_i^-$ is the is the rate at which a motor pulling in the positive direction is replaced by one moving in the negative direction, $\kappa_0$ is a force-free transition rate, and $\gamma$ is a  dimensionless load sensitivity defined as twice the stall force divided by the detachment force scale~\citep{allard2019sliding}. The relative velocity of the parallel microtubules in each state is given by
\begin{align}\label{eq:ki_rates_sliding}
\Delta v_i &= V_m \frac{2i-K}{K} \,, 
\end{align}
where $V_m$ is the speed of a single motor \citep{allard2019sliding}. 

A main point of this study is that parallel microtubules may slide bidirectionally with respect to each other, with a zero mean velocity due to symmetry, thus the long-term microtubule transport is characterized by diffusive behavior. The effective diffusivity of the microtubule pair driven by a total of $35$ motors is measured in \citet{allard2019sliding} through fitting the slope of the mean squared displacement in stochastic simulations at long time (stars in Figure~\ref{fig:mt_sliding}) and compared to a theoretical approximation in terms of a first passage time problem (open circles in Figure~\ref{fig:mt_sliding}).

Our method based on renewal rewards theory yields predictions of the effective diffusivity that are closer to the estimates derived from large-time stochastic simulations (marked with triangles in Figure~\ref{fig:mt_sliding}). To make the relationship between effective diffusivity and load sensitivity more clear, we illustrate results for many intermediary values.  Our proposed analytical framework also facilitates the possibility of subsequent systematic asymptotic approximations to study dependence on underlying biophysical parameters.

\section{Discussion}
\label{sec:conclusions}

In this work, we consider examples from the intracellular transport literature where particles undergo switching dynamics. In particular, we are interested in determining the effective velocity and diffusivity as well as the expected run length of these particles as they switch between biophysical behaviors such as diffusion, active transport, and stationary states. We propose a method that is based on defining the underlying Markov chain of state switches and the independent cycles of the dynamics marked by returns to a chosen base state. Emphasizing the cyclic structure of the behavior allows us to treat the time durations and spatial displacements of particles in these regenerative cycles as the cycle durations and rewards in a renewal-reward process. Through calculation of the statistics of cycle time and displacement, this robust framework provides a rigorous means to study how the dynamics of switching systems depends on model parameters.

This approach applies for example to canonical tug-of-war models describing the transport of cargo by teams of molecular motor proteins. Previous investigations of the effective transport of cargo in these multi-state models have considered individual trajectories of the dynamics, computed using Monte Carlo simulations with the Gillespie algorithm~\citep{muller2008tug,kunwar2010robust,muller2010bidirectional}. These studies determine the effective velocity of the particles analytically by calculating the distribution of the number of bound motors from the stationary solution of the master equation \citep{klumpp2005cooperative}. However, determining the effective diffusivity in these studies relied on numerical simulations. Our method proposes a faster and explicit investigation of the impact of model parameters on the effective diffusivity. For instance, Figure~\ref{fig:tugwar_kin} (top right) captures the different behavior of identical motor teams involved in tug-of-war dynamics when the ratio of stall to detachment force is small (weak motors with small effective diffusivity) versus large (strong motors with increasing effective diffusivity). This observation is consistent with simulations in \citet{muller2008tug}, where the large force ratios correspond to a dynamic instability where only one motor type is primarily bound at the end of an unbinding cascade \citep{muller2008tug,muller2010bidirectional}.

Multiple experiments summarized in \citet{hancock2014bidirectional} have shown that inhibition of one motor type reduces transport in both directions in several systems, suggesting a ``paradox of co-dependence'' in bidirectional cargo transport. Several mechanisms accounting for this paradox were proposed, including the microtubule tethering mechanism recently explored in \citet{smith2018assessing}. The hypothesis for this mechanism is that motors switch between directed active transport and a weak binding or diffusive state. The recent experimental study in \citet{feng2018motor} suggests that teams of kinesin-1 motors coordinate transport using help from the dynamic tethering of kinesin-2 motors. This work shows that when kinesin-1 motors detach, tethering of kinesin-2 to the microtubule ensures that cargo stays near the filament to allow for subsequent reattachment \citep{feng2018motor}. Our approach allows us to assess the dependence of the dynamics on a potentially increased reattachment rate for cargo that is already bound to the filament by at least one motor (Figure~\ref{fig:tugwar_rho}). Implementing this change in the standard binding models in \citet{muller2008tug,kunwar2010robust,muller2010bidirectional} for both kinesin-1/dynein and kinesin-1/DDB dynamics shows a decrease in overall effective velocity, but very large increases in potential run length. This could be consistent with the paradox in that experimentalists would observe more kinesin-directed activity when the reattachment rate is sufficiently high.

We have made Matlab and Mathematica sample code available for the calculation of effective velocity, diffusivity, and run lengths in a cooperative model of one-directional transport (as discussed in \ref{sec:cooperative}) and a tug-of-war model of bidirectional transport (as discussed in \ref{sec:tugwar}) \citep{Github_effdiff2019}. The code for these examples can be readily adapted to allow for a general probability transition matrix for the state dynamics, together with the probability distributions for the times and displacement in each state, to extend to other models of the processive movement of molecular motors and cargo transport. As the theory of how motors coordinate to transport cargo continues to develop at a rapid pace, the analysis developed here will provide a tool for new models accounting for tethered and weakly binding states with stochastic transitions whose rates do not depend on spatial position.

%% Bibliography
%\newpage
\bibliographystyle{spbasic} 
\bibliography{references}   % name your BibTeX data base

%\paragraph{Paragraph headings} Use paragraph headings as needed.

%\begin{acknowledgements}
%\end{acknowledgements}

%\begin{thebibliography}{}
%%
%% and use \bibitem to create references. Consult the Instructions
%% for authors for reference list style.
%\end{thebibliography}

\end{document}